\documentclass[10pt,a4paper]{article}
\usepackage[T1]{fontenc}
\usepackage[left=2cm, right=2cm, top=3cm, bottom=1.5cm]{geometry}
\usepackage{graphicx}
\usepackage{mathtools}
\usepackage{amssymb}
\usepackage{amsthm}
\usepackage{hyperref}
\usepackage{multirow}
\usepackage{float}
\usepackage{subfigure}
\usepackage{threeparttable}
\usepackage{authblk}

\usepackage{amsmath}
\usepackage{rotating}
\usepackage{times}
\usepackage{bm}
\usepackage{natbib}
\usepackage{hyperref}
\usepackage[plain,noend]{algorithm2e}

\providecommand{\keywords}[1]
{
	\small	
	\quad \textbf{Keywords:} #1
}

\newtheorem{theorem}{Theorem}
\newtheorem{corollary}{Corollary}
\newtheorem{lemma}{Lemma}
\newtheorem{assumption}{Assumption}

\newtheorem{definition}{Definition}
\newtheorem{example}{Example}
\newtheorem{proposition}{Proposition}

\newcommand\blfootnote[1]{%
	\begingroup
	\renewcommand\thefootnote{}\footnote{#1}%
	\addtocounter{footnote}{0}%
	\endgroup
}

\title{Improve Sensitivity Analysis Synthesizing Randomized Clinical Trials With Limited Overlap
\blfootnote{Kuan Jiang and Wenjie Hu have equal contributions to this work. 
}
}

\author[a]{Kuan Jiang}
\author[b]{Wenjie Hu}
\author[c]{Shu Yang}
\author[d]{Xinxing Lai} 
\author[e]{Xiaohua Zhou \thanks{Corresponding author}}

\affil[a]{Department of Biostatistics, School of Public Health, Peking University }
\affil[b]{Department of Probability and Statistics, Peking University}
\affil[c]{Department of Statistics, North Carolina State University}
\affil[d]{Institute for Brain Disorders, Beijing University of Chinese Medicine}
\affil[e]{Department of Biostatistics and Beijing International Center for Mathematical Research, Peking University}

\begin{document}

	\date{}
	\maketitle

\begin{abstract}
Randomized clinical trials are the gold standard when estimating the average treatment effect. However, they are usually not a random sample from the real-world population because of the inclusion/exclusion rules.    
Meanwhile, observational studies typically consist of representative samples from the real-world population. 
However, due to unmeasured confounding, sensitivity analysis is often used to estimate bounds for the average treatment effect without relying on stringent assumptions  of other existing methods. This article introduces a synthesis estimator that improves sensitivity analysis in observational studies by incorporating randomized clinical trial data, even when overlap in covariate distribution is limited due to inclusion/exclusion criteria. We show that the proposed estimator will give a tighter bound when a “separability” condition holds for the sensitivity parameter. Theoretical proofs and simulations show that this method provides a tighter bound than the sensitivity analysis using only observational study. We apply this method to combine an observational study on drug effectiveness with a partially overlapping RCT dataset, yielding improved average treatment effect bounds.

\end{abstract}


\section{Introduction}
To estimate the average treatment effect (ATE), the randomized clinical trial (RCT)  offers unbiased estimation due to perfect randomization and is considered as the gold standard. However, the RCT population may differ from the real-world population due to the covariates distribution shift from the sample selection procedure, statistical methods are necessary to adjust the bias. Existing RCT generalization methods 
can be categorized as follows: outcome model-based method \citep{dahabreh2019generalizing}, inversed probability of sampling weighting method \citep{cole2010generalizing,stuart2011use,tipton2013improving,o2014generalizing,buchanan2018generalizing}, augmented inversed probability of sampling weighting method \citep{zhang2016new,dahabreh2019generalizing}, calibration weighting method and augmented calibration weighting method \citep{lee2023improving}.  
These methods assume a positivity assumption, which requires that all real-world population patients have a positive probability of RCT selection.  
However, because of inclusion and exclusion criteria, the positivity assumption may not be satisfied for the RCT-generalization methods, so the adjusted ATE represents only the  real-world trial-eligible population \citep{zivich2022positivity}. Recently, \citet{zivich2022positivity,zivich2023synthesis} attempted to overcome violations of this assumption by combining a statistical generalization model and a parametric model to simulate the ATE. However, their model relies on parameters chosen based on domain knowledge, which may not always be feasible in practice. 

Observational studies (OS) on real-world cohorts can be used to estimate the ATE for the real-world population when the RCT is not available.
However, consistent ATE estimation through OSs relies on the no unmeasured confounders assumption, which can be violated due to insufficient controlled confounders \citep{Imbens_Rubin_2015}. While methods like utilizing instrumental variables \citep{imbens1994identification,angrist1996identification,newey2003instrumental} and double negative control variables \citep{miao2018confounding,tchetgen2020introduction,miao2023identifying} aim to address this issue, they require additional assumptions that are difficult to verify in practice. Therefore, sensitivity analysis is widely used to provide bounds on the ATE without relying on complex assumptions. The key idea for conducting sensitivity analysis on ATE bounds is adjusting the inconsistent estimator with non-identified sensitivity parameters to obtain a consistent estimator for ATE. Researchers set a feasible range for the sensitivity parameter based on their experience and then infer the ATE bounds accordingly. So far, plenty of methods have been proposed by researchers such as
\citet{rosenbaum1983assessing,robins1999association,frank2000impact,imbens2003sensitivity,brumback2004sensitivity,rosenbaum2005sensitivity,frank2008does,hosman2010sensitivity,imai2010identification,vanderweele2011bias,blackwell2014selection,dorie2016flexible,franks2019flexible,oster2019unobservable,zhao2019sensitivity,cinelli2020making,chernozhukov2022long,lu2023flexible,zhou2023sensitivity,dorn2023sharp}. Specifically, \citet{cinelli2020making,chernozhukov2022long} propose a methodology based on omitted variable bias to evaluate the robustness of the estimation when there exists potential unmeasured confounders. The method can be conveniently applied to practical analysis because benchmark variables can be used to bound the unmeasured confounding strength.

In this article, we give a new sensitivity analysis framework for estimating the ATE in the  real-world population by combing OS dataset and RCT dataset that have limited covariates distribution overlap due to inclusion/exclusion criteria. Under this framework, we propose a synthesis estimator that combines RCT-generalization methods and classical sensitivity analysis methods. 
This synthetic estimator  addresses the positivity assumption violations in existing RCT generalization methods, and gives a valid bound for the ATE in the super-population. 
Besides, we show that the synthesis estimator can give tighter bounds for the ATE comparing to the classical sensitivity analysis methods using only OS data when a “separability” condition holds for the sensitivity parameter.   Compared to \cite{zivich2022positivity,zivich2023synthesis}, we do not require a parametric model in our synthesis estimator, which makes our approach more robust against possible model misspecifications.
We also conduct simulations and real data analysis to demonstrate the good performance of the synthesis estimator compared to the sensitivity analysis method using only OS data.   
All proofs are relegated to the supplementary materials.

\section{Notations and Assumptions}\label{sec:2}

To estimate the ATE in a super-population, denoted as $\psi$, we consider covariates $X \in \mathbb{R}^{q}$, a binary treatment assignment $A = a, a \in \mathcal{A}, \mathcal{A} = \{0,1\}$, and the outcome $Y$. Using the potential outcome framework by \citet{Imbens_Rubin_2015}, each unit has potential outcomes \(Y(1)\) and \(Y(0)\) corresponding to treatments \(A = 1\) and \(A = 0\), respectively. We assume that $Y = AY(1) + (1 - A)Y(0)$. Our target parameter is represented as \(\psi = E\{Y(1) - Y(0)\}\).

The observational study cohort is a random sample from the real-world population without selection bias, while the RCT sample depends on inclusion/exclusion criteria and complicated sampling procedures, leading to an selection bias. We denote the observational and RCT datasets as \(D_O\) and \(D_R\), respectively. Assume the covariates \(X\) can be split into two parts: \(X = (W, V)\), where \(W \in \mathbb{R}^s\) are covariates independent of inclusion/exclusion restrictions, and \(V \in \mathbb{R}^t\) are covariates restricted by the known inclusion/exclusion criteria (\(s + t = q\)). A binary variable \(V^*\) indicates whether a unit violates the RCT's inclusion/exclusion criteria. Units with \(V_i^* = 0\) can participate in the RCT, while those with \(V_i^* = 1\) are excluded. For example, if patients over 60 years old are excluded, then \(V^*_i = I(\text{age}_i > 60)\). Let \(S_i = 1\) denote RCT participation. The RCT dataset can be represented as \(D_R = \{A_i, W_i, V_i, Y_i, V^*_i = 0, S_i = 1\}, i \in \{1, ..., n\}\). The observational dataset \(D_O\) can be split into two subsets: \(\{A_i, X_i, Y_i, V^*_i = 0, S_i = 0\}, i \in \{n+1, ..., n+N_0\}\) and \(\{A_i, X_i, Y_i, V^*_i = 1, S_i = 0\}, i \in \{n+N_0+1, ..., n+N_0+N_1\}\), where \(N = N_0 + N_1\). The former overlaps completely with the RCT dataset in terms of covariates, while the latter does not.

Similar to \(\psi\), we define \(\psi_0 = E \{Y(1) - Y(0) \mid V^* = 0\}\) as the ATE for the sub-population with \(V_i^* = 0\), and define \(\psi_1\) similarly for the sub-population with \(V_i^* = 1\). 

In the RCT dataset, we assume the following standard assumptions:

\begin{assumption}{(Ignorability in RCT)}\label{ass:1}
	$E\{Y(a) \mid X,S = 1, A = a\} = E\{Y(a) \mid X, S = 1\}$, for every $a \in \mathcal{A}$.
\end{assumption}

\begin{assumption}{(Positivity of treatment assignment probability in RCT)}\label{ass:2}
	$pr(A = a \mid X = x, S = 1) > 0$, for every $a \in \mathcal{A}$.
\end{assumption}

Assumption \ref{ass:1}  is satisfied under randomization, and  Assumption \ref{ass:2} requires each unit to have a positive probability in the treatment and control group. In addition, to utilize the RCT generalization methods and sensitivity analysis methods, we assume the mean generalizability assumption.

\begin{assumption}\label{ass:3}
	(Mean generalizability on common support) $E\{Y(a)\mid W = w, V = v, V^* = 0, S = 1\} = E\{Y(a)\mid W = w, V = v, V^* = 0, S = 0\}$, for every $a \in \mathcal{A}$.
\end{assumption} 


Assumption \ref{ass:3} links the RCT population to the real-world trial-eligible population for generalization methods. It assumes that the conditional potential outcomes given covariates are identical between these populations.

\section{Examples of Sensitivity Analysis Methods in Observational Studies} \label{sec:3}

The section aims to introduce some existing sensitivity analysis methods, which will be used as examples in the following sections. 

\begin{example}{\citep{lu2023flexible}}\label{ex:1}
	To conduct a sensitivity analysis on the ATE, the sensitivity parameters are defined as:
	\begin{equation*}
		\epsilon^1(X) = \frac{E\{Y(1) \mid A = 1, X\}}{E\{Y(1)\mid A = 0, X\}} , \quad \epsilon^0(X) = \frac{E\{Y(0)\mid A = 1, X\}}{E\{Y(0)\mid A = 0, X\}}.
	\end{equation*}
	
	These sensitivity parameters range in the $\mathcal{L}^\infty$ space. If $\epsilon^1(X) = \epsilon^0(X) \equiv 1$, no unmeasured confounders are assumed. In practice, even if the sensitivity parameters are functions of $X$, they can be simplified to a square in $\mathbb{R} \times \mathbb{R}$ containing $\epsilon_1 = \epsilon_0 = 1$. Essentially, this involves choosing a class of constant functions that include $1$.
	
	Given $\epsilon^1(X)$ and $\epsilon^0(X)$, we can adjust existing ATE estimators to consistently account for unmeasured confounders. The modified formulas are presented in Table \ref{tab:1}.

\begin{table}
	\centering
	\begin{threeparttable}
		\caption{Outcome regression estimator, inverse probability weighting estimator, augmented inverse probability weighting estimator and their modified forms in \citet{lu2023flexible} }
		\label{tab:1}
		\begin{tabular}{ccc}
			Estimator & Version & Formula \\
			OR   & Original &   $\frac{1}{N}\sum_{i = 1}^{N}\{\hat{\mu}_1(X) - \hat{\mu}_0(X)\} $  \\
			& Modified & $\frac{1}{N}\sum_{i = 1}^{N}\left\{A Y+(1-A)\frac{\hat{\mu}_{1}(X)}{\varepsilon_{1}(X)}\right\}-\frac{1}{N}\sum_{i = 1}^{N} \left\{A \hat{\mu}_{0}(X) \varepsilon_{0}(X)+(1-A) Y\right\}$ \\
			IPW &  Original  & $\frac{1}{N}\sum_{i = 1}^{N}\left\{\frac{A}{\hat{e}(X)} Y  - \frac{1 - A}{1 - \hat{e}(X)}Y\right\}$  \\
			& Modified  &   $\frac{1}{N}\sum_{i = 1}^{N} \left\{\hat{w}_{1}(X) \frac{A}{\hat{e}(X)} Y\right\} - \frac{1}{N}\sum_{i = 1}^{N} \left\{\hat{w}_0(X)\frac{1 - A}{1 - \hat{e}(X)}Y\right\}$  \\
			AIPW	 &   Original   & $\frac{1}{N}\sum_{i = 1}^{N} \left\{\frac{AY}{\hat{e}(X)} - \frac{\{A - \hat{e}(X)\}\hat{\mu}_1(X)}{\hat{e}(X)} \right\} - \frac{1}{N}\sum_{i = 1}^{N} \left[\frac{(1 - A)Y}{1 - \hat{e}(X)} - \frac{\{\hat{e}(X) - A\}\hat{\mu}_0(X)}{1 - \hat{e}(X)} \right] $  \\
			& Modified & $\frac{1}{N}\sum_{i = 1}^{N}\left\{\hat{w}_1(X)\frac{AY}{\hat{e}(X)} - \frac{\{A - \hat{e}(X)\}\hat{\mu}_1(X)}{\hat{e}(X)\varepsilon_1(X)} \right\}- \frac{1}{N}\sum_{i = 1}^{N} \left[\hat{w}_0(X)\frac{(1 - A)Y}{1 - \hat{e}(X)} - \frac{\{A - \hat{e}(X)\}\hat{\mu}_0(X)\varepsilon_0(X)}{1 - \hat{e}(X)} \right]$\\
		\end{tabular}
	\end{threeparttable}
	\begin{tablenotes}
		\item[]OR, outcome regression; IPW, inverse probability weighting estimator; AIPW, augmented inverse probability weighting estimator.
		\item[]$\hat{\mu}_{a}, a \in \{0,1\}$ is the outcome model fitted on the group with $A = a$.
		\item[]$\hat{w}_1(X) = \hat{e}(X) + {(1 - \hat{e}(X))}/{\epsilon_1(X)}, \hat{w}_0(X) = \hat{e}(X)\epsilon_0(X) + 1 - \hat{e}(X)$.
	\end{tablenotes}
\end{table}
\end{example}

\begin{example}{\citep{cinelli2020making}} \label{ex:2}
	The framework uses two partial $R^2$ parameters to evaluate the bias of the confounded estimation. The sensitivity analysis parameters are defined as:
	\begin{equation*}
		R^2_{Y \sim U\mid A, X} = \frac{R^2_{Y \sim A+X+U} - R^2_{Y \sim A + X}}{1 - R^2_{Y \sim A + X}}, \quad R^2_{A \sim U\mid X} = \frac{R^2_{A \sim X+U} - R^2_{A \sim  X}}{1 - R^2_{A \sim X}} ,  
	\end{equation*}
	where $U$ denotes the unmeasured confounder, and the $R^2$ values on the right-hand side of the equations are used in linear regression. The advantage of this framework is that the sensitivity parameter is a coordinate in $[0,1] \times [0,1)$ and is inherently bounded. Thus, in the linear model, with a chosen pair of parameters, the absolute value of the bias of the estimated ATE can be represented as the following formula:
	\begin{equation*}
		|\widehat{bias}(R^2_{A \sim U|X}, R^2_{Y \sim U|A, X})| = se(\hat{\psi}_{res}) \left(\frac{R^2_{Y \sim U\mid A, X}R^2_{A \sim U\mid X} }{1 - R^2_{A \sim U\mid X}} df\right)^{1/2}, 
	\end{equation*}
	where $\widehat{bias}(\cdot,\cdot)$ is the bias function of the sensitivity parameters,  $\hat{\psi}_{res}$ is the linear regression estimator omitting the unmeasured confounders,  $se(\hat{\psi}_{res})$ is the standard error of the estimator, and $df$ is degree of freedom. With the confounded ATE and the bias, the true ATE can be obtained by adjusting the confounded ATE with the bias (plus or minus the absolute value of bias). 
	
	In \citet{chernozhukov2022long}, this framework is extended to non-linear models, following a procedure similar to the linear case. For generality, the subsequent analysis will primarily focus on the partial $R^2$ sensitivity parameters proposed in \citet{chernozhukov2022long} and can refer to Example \ref{ex:4} .
\end{example}

From the examples above, we introduce some additional notations in the general sensitivity analysis framework. For a specific sensitivity analysis method implemented on the OS dataset, denote $\epsilon$ as the sensitivity analysis parameter. We denote $\psi(\epsilon)$ as a function of $\epsilon$. So when $\epsilon$ is correctly specified, we have $\psi = \psi(\epsilon)$. Therefore, the estimand $\psi$ can be represented as $\psi(\epsilon)$. Similarly, for the sub-dataset of the OS data with $V^* = j, j \in \{0,1\}$, we can denote $\epsilon_j$ as the corresponding sensitivity parameters, and we define  $\psi_j(\epsilon_j)$ in the same way as $\psi(\epsilon)$.

\section{Methodology}

\subsection{A Synthesis Formulation}
Our parameter of interest is the ATE in the real-world population $\psi$. By the law of total expectation, we can separate $\psi$ as:
\begin{equation}\label{eq:1}
	\psi = p_0\psi_0  + p_1\psi_1,
\end{equation}
where $p_0 = pr(V^* = 0)$  and $p_1 = 1-p_0 = pr(V^* = 1)$  represent the probabilities of a unit meeting the inclusion/exclusion criteria, respectively. And we call $p_0$ as RCT-overlapping proportion. Because the OS samples are drawn randomly from the real-world population, these probabilities can be approximated by \( N_0/N \) and \( N_1/N \), denoted as \( p_0^N \) and \( p_1^N \).

The above equation motivates us to integrate  RCT data and OS data by combining RCT generalization techniques and sensitivity analysis to estimate the ATE for the real-world population. 
Specifically, we employ RCT generalization methods to estimate \( \psi_0 \) and use sensitivity analysis to determine bounds for \( \psi_1 \).  Denote $\hat{\psi}_{gen}$ as the RCT generalization estimator which extends the RCT results to the real-world trial-eligible population and $\hat{\psi}_{1}(\epsilon_1)$ as the estimator for $\psi_1(\epsilon_1)$. We   combine these two estimators and get the synthesis sensitivity analysis estimator $\hat{\psi}_{syn}(\epsilon_1)$:
\begin{equation*}
	\hat{\psi}_{syn}(\epsilon_1) = p_0^N\hat{\psi}_{gen} + p_1^N \hat{\psi}_{1}(\epsilon_1).
\end{equation*}
Similar to $\hat{\psi}_{1}(\epsilon_1)$, the corresponding sensitivity analysis estimators for $\psi_0(\epsilon_0)$ and $\psi(\epsilon)$ are represented as \( \hat{\psi}_0(\epsilon_0) \) and \( \hat{\psi}(\epsilon) \). Depending on the specific RCT generalization method chosen for practical analysis, $\hat{\psi}_{gen}$ can have different expressions. We review several RCT generalization methods, which are summarized in the following Table \ref{tab:2} for clarity.

\begin{table}[h]
	\centering
	\begin{threeparttable}
		\caption{Examples of RCT generalization methods and corresponding expressions of $\hat{\psi}_{gen}$}
		\label{tab:2}
		\begin{tabular}{cc}
			RCT generalization method	&  $\hat{\psi}_{gen}$\\
			OM	&  $\frac{1}{N_0}\sum_{i = n+1}^{n+N_0}\{\hat{\mu}_{1,1}(X_i) - \hat{\mu}_{0,1}(X_i)\}$ \\
			IPSW & $\frac{1}{n}\sum_{i = 1}^{n} \frac{n}{N_0}\frac{Y_i}{\hat{\alpha}(X_i)}\left\{\frac{A_i}{\hat{e}(X_i)} - \frac{1 - A_1}{1 - \hat{e}(X_i)}\right\}$ \\
			AIPSW& $\frac{1}{n}\sum_{i = 1}^{n} \frac{n}{N_0}\frac{1}{\hat{\alpha}(X_i)}\left[\frac{A_i\{Y_i - \hat{\mu}_{1,1}(X_i)\}}{\hat{e}(X_i)} - \frac{(1 - A_1)\{Y_i - \hat{\mu}_{0,1}(X_i)\}}{1 - \hat{e}(X_i)}\right] $ \\
			& $+ \frac{1}{N_0}\sum_{i = n+1}^{n+N_0}\{\hat{\mu}_{1,1}(X_i) - \hat{\mu}_{0,1}(X_i)\}$ \\
		\end{tabular}
		\begin{tablenotes}
			\item[] OM, outcome model-based estimator; IPSW, inversed probability of sampling weighting estimator; AIPSW, augmented inversed probability of sampling weighting estimator.
			\item[] $\hat{\mu}_{a,1}, a\in \{0,1\}$ is the outcome model fitted on the RCT data.
			\item[] $\hat{\alpha}(X)$ is an estimate of the odds of the probability of being in the RCT.  
			\item[] $\hat{e}(X)$ is the estimated propensity score in the RCT. If the patients follow perfect randomization, then $\hat{e}(X)$ is known.
		\end{tablenotes}
	\end{threeparttable}
\end{table}

First, we will discuss the consistency and asymptotic normality of the synthesis estimator with a specified sensitivity parameter in the following theorem.

\begin{theorem}\label{thm:1}
	For a specified $\epsilon_1$, assume that $n/N \rightarrow \rho$, where $\rho$ is a constant.  Under Assumptions \ref{ass:1} - \ref{ass:3}, $\hat{\psi}_{syn}(\epsilon_1)$ has the following properties as n $\rightarrow$ $\infty$:
	
	(i). \textbf{Consistency} $\hat{\psi}_{syn}(\epsilon_1) \to \psi(\epsilon_1)$ in probability. Especially, if $\epsilon_1$ is correctly specified, i.e. $\psi_1(\epsilon_1) = \psi_1$, then $\hat{\psi}_{syn}(\epsilon_1) \to \psi$ in probability.
	
	
	(ii). \textbf{Asymptotic Normality} For a specified $\epsilon_1$, $ {N}^{1/2} (\hat{\psi}_{syn} - \psi) \to \mathcal{N}[0, p_0(1-p_0)\{\psi_0-\psi_1(\epsilon_1)\}^2 + p_0\sigma_0^2 + p_1 \sigma_1^2]$ in distribution, where $N_0^{1/2}(\hat{\psi}_{gen} - \psi_0) \rightarrow N(0, \sigma_0^2)$ and $N_1^{1/2}(\hat{\psi}_1(\epsilon_1) - \psi_0) \rightarrow N(0, \sigma_1^2)$.
\end{theorem}

The proof of consistency in Theorem \ref{thm:1} is straightforward. According to the RCT generalization method, $\hat{\psi}_{gen}$ is consistent for $\psi$ if the models for nuisance parameters are correctly specified, i.e., $\hat{\psi}_{gen} \rightarrow \psi_0$. For $\hat{\psi}_{1}(\epsilon_1)$, with correctly specified $\epsilon_1$ (denoted as $\epsilon_1^{*}$), we have $ \hat{\psi}_1(\epsilon_1^{*}) \rightarrow \psi_1$. Thus, $\hat{\psi}_{syn}(\epsilon_1^*) \rightarrow \psi$ based on Equation \eqref{eq:1}. 
The asymptotic property in the theorem is proved in the supplemental material.

\subsection{Decomposition of Sensitivity Parameters}

The synthesis method prompts us to explore the relationship between sensitivity parameters for the super-population and its sub-populations. If we can express the  sensitivity parameter for the super-population as a function of the parameters on sub-populations, we can impose additional constraints on the sensitivity parameter for the RCT-overlapping sub-population using $\hat{\psi}_{gen}$. This relationship, referred to as “separability”, is introduced in the following definition.

\begin{definition}{(Separable Sensitivity Parameter)}\label{def:1}
	Suppose $\epsilon$ is the sensitivity parameter for the super-population, and $\epsilon_j, j \in \{0,1\}$ is the sensitivity parameter for the sub-population with $V^\ast = j$. We assume $\epsilon_0$ and $\epsilon_1$ are variationally independent and denote $\mathcal{E}, \mathcal{E}_0, \mathcal{E}_1$ as the sensitivity parameter space for $\epsilon, \epsilon_0, \epsilon_1$, respectively. 
	The sensitivity parameter $\epsilon$ is called separable if for any $\epsilon \in \mathcal{E}$, there exists a function $g$ such that:
	\begin{equation}\label{eq:2}
		\epsilon = g(\epsilon_0,\epsilon_1;\theta),  
	\end{equation}
	for $\epsilon_j \in \mathcal{E}_j$, where $\theta$ is a non-identifiable nuisance parameter that is variationally independent of the sensitivity parameters $\epsilon_0, \epsilon_1$. The function $g$ is called the transformation function.  So $\mathcal{E} = \{ \epsilon\mid \epsilon = g(\epsilon_0,\epsilon_1;\theta), \epsilon_0 \in \mathcal{E}_0, \epsilon_1 \in \mathcal{E}_1, \theta \in \mathcal{M} \}$. 
\end{definition}

Under a nonparametric model for the super-population, $\epsilon_0$ and $\epsilon_1$ are defined in two disjoint sub-populations, therefore they are variationally independent. However, when the super-population has some structural restrictions, $\epsilon_0$ and $\epsilon_1$ may not be variationally independent, then Definition \ref{def:1} is not satisfied. 
The existence of non-identifiable nuisance parameter $\theta$ may complicate our sensitivity analysis. 
So it's necessary to verify that the sub-population parameters $\epsilon_0, \epsilon_1$ and the nuisance parameter $\theta$ are variationally independent, which means that $\theta$ has the same support when $\epsilon_0, \epsilon_1$ change. When the separability assumption holds, $\theta$ can be treated as a fixed value while varying the sub-population sensitivity parameters $\epsilon_0, \epsilon_1$ within their spaces. 


The transformation function in Equation \eqref{eq:2} also depends on the proportion of each sub-population. Therefore, it is better to revise the notation $g$ in Equation \eqref{eq:2} to $g_p$, where $p = (p_0, p_1)$ is a vector representing the proportions of the sub-populations. Then, it is trivial that:
\begin{equation*}
	\lim_{p_j \to 1} \epsilon = \lim_{p_j \to 1} g_p(\epsilon_0, \epsilon_1; \theta) = \epsilon_j, \quad j \in \{0,1\}.
\end{equation*}


\subsection{Examples of Separable Sensitivity Analysis Parameters}

The following two examples show that the sensitivity parameters in the sensitivity analysis methods in Section \ref{sec:3} are both separable.
\begin{example}{(Example \ref{ex:1}, continued)}\label{ex:3}
	Since the sensitivity parameters are functions of $X$ only, and are separated by $V^*$, which is an indicator dependent on $X$, we have:
	\begin{equation*}
		\epsilon^1(X) = \epsilon_0^1(X)I(V^* = 0) + \epsilon_1^1(X)I(V^* = 1), \quad  \epsilon^0(X) = \epsilon_0^0(X)I(V^* = 0) + \epsilon_1^0(X)I(V^* = 1),
	\end{equation*}
	where $\epsilon_j^1(X)$ and $\epsilon_j^0(X), j \in \{0,1\}$ are sensitivity parameters (functions) on the sub-populations with $V^* = j$, respectively. Since the two sub-populations are mutually disjoint, the two parameters are separable mutually. In particular, there is no nuisance parameter $\theta$ in this example, which simplifies the verification of Definition \ref{def:1}. 
\end{example}

\begin{example}{(Example \ref{ex:2}, continued)}\label{ex:4}
	We use another definition of $R^2_{Y \sim U\mid A, X}$ and $R^2_{A \sim U\mid X}$ proposed in \cite{chernozhukov2022long} to separate the sensitivity parameters in Example \ref{ex:2}. Denote  $R^2_{Y \sim U\mid A, X,V^* = j}$ and $R^2_{A \sim U\mid  X, V^* = j}$, $j \in \{0,1\}$ as sensitivity parameters on the sub-population with $V^\ast = j$, and we have the following results:
	\begin{equation*}
		\begin{aligned}
			R^2_{Y \sim U\mid A, X} & = \frac{var\{E(Y\mid U,A,X)\} - var\{E(Y\mid A,X)\} }{var(Y) - var\{E(Y\mid A,X)\}}  \\
			& = \frac{\sum_{j \in \{0,1\}} p_j\frac{R^2_{Y \sim U\mid A, X,V^* = j}}{\Delta_{yj}} + \frac{var[E\{E(Y\mid U,A,X)\mid V^*\}] - var[E\{E(Y\mid A,X)\mid V^*\}]}{\prod_{j \in \{0,1\}}\Delta_{yj}}}{\sum_{j \in \{0,1\}} p_j\frac{1}{\Delta_{yj}} + \frac{var\{E(Y\mid V^*)\} - var[E\{E(Y\mid A,X)\mid V^*\}]}{\prod_{j \in \{0,1\}}\Delta_{yj}}},
		\end{aligned}
	\end{equation*}
	\begin{equation*}
		\begin{aligned}
			R^2_{A \sim U\mid  X} & = \frac{var\{E(A\mid U,X)\} - var\{E(A\mid X)\} }{var(A) - var\{E(A\mid X)\}}  \\
			& = \frac{\sum_{j \in \{0,1\}} p_j\frac{R^2_{A \sim U\mid X,V^* = j}}{\Delta_{aj}} + \frac{var[E\{E(A\mid U,X)\mid V^*\}] - var[E\{E(A\mid X) \mid V^*\}]}{\prod_{j \in \{0,1\}}\Delta_{aj}}}{\sum_{j \in \{0,1\}} p_j\frac{1}{\Delta_{aj}} + \frac{var\{E(A\mid V^*)\} - var[E\{E(A\mid X)\mid V^*\}]}{\prod_{j \in \{0,1\}}\Delta_{aj}}},
		\end{aligned}
	\end{equation*}
	where $\Delta_{yj} = var(Y\mid  V^* = 1- j) -  var\{E(Y\mid A,X)\mid  V^* = 1- j\}$ and $\Delta_{aj} = var(A\mid  V^* = 1- j) -  var\{E(A\mid X)\mid  V^* = 1 - j\}$.  However, there is an unknown nuisance parameter $var[E\{E(Y\mid U,A,X)\mid V^*\}]$, which is not identifiable if $U$ is unknown.
\end{example}

The two examples illustrate two cases where the nuisance parameter exists or not.
In the former case, the sensitivity parameter is a function of $X$, 
allowing it to be separated by the indicator function $V^*$ for the sub-populations. 
A notable feature of this parameter is that there is a one-to-one correspondence between the sensitivity parameter on the super-population and those on the sub-populations. 
However, in the latter case, this property does not hold; the sensitivity parameters on the sub-populations cannot be directly identified from the sensitivity parameters on the super-population due to the existence of nuisance parameter.

\subsection{Sensitivity Analysis Bounds}

In this section, we aim to compare the sensitivity analysis method using only OS data with the synthesis sensitivity analysis method in Section \ref{sec:2} when the sensitivity parameter satisfies Definition \ref{def:1}.
For a sensitivity analysis method for ATE, the bound width is defined as the width of ATE bound when the sensitivity parameters vary within their spaces, that is,
\begin{equation*}
	W = \sup_{\epsilon \in \mathcal{E}} \hat{\psi}(\epsilon) - \inf_{\epsilon \in \mathcal{E}} \hat{\psi}(\epsilon).
\end{equation*}

When RCT data is available, it is natural to improve the sensitivity analysis with only OS data by adding restrictions on the sensitivity parameter in RCT-overlapping sub-population, which induces the RCT-enhanced sensitivity analysis estimator defined as:
\begin{equation*}
	\hat{\psi} \circ g_p(\epsilon_0, \epsilon_1; \theta), \quad \text{subject to} \quad \hat{\psi}_0(\epsilon_0) = \hat{\psi}_{gen}.
\end{equation*}

Moreover, denote $\epsilon_0^*$ as the solution to the restriction above, this estimator can be represented as $\hat{\psi }\circ g_p(\epsilon_0^*, \epsilon_1; \theta)$.  For the sensitivity analysis method whose sensitivity parameters satisfy Definition \ref{def:1}, we can express the estimator as $\hat{\psi} \circ g_p(\epsilon_0, \epsilon_1; \theta)$, where $(\epsilon_0, \epsilon_1) \in \mathcal{E}_0 \times \mathcal{E}_1$. However, for the RCT-enhanced method, the coordinate $(\epsilon_0^*, \epsilon_1)$ is in the space $ \mathcal{E}_0^{'} \times \mathcal{E}_1$, where $\mathcal{E}_0^{'} \subset \mathcal{E}_0$. Therefore, we introduce the following lemma comparing the bounds between the sensitivity analysis method with only OS data and the corresponding RCT-enhanced estimator. The lemma is stated as follows:

\begin{lemma}\label{lem:1}
	Suppose the sensitivity parameter $\epsilon$ for the super-population satisfies Definition \ref{def:1}. For the  sensitivity analysis estimator $\hat{\psi}(\epsilon)$ and the RCT-enhanced sensitivity analysis estimator $\hat{\psi}\circ g_p(\epsilon_0^*, \epsilon_1; \theta)$, the following conclusions hold:
	
	(i).$\inf_{\epsilon \in \mathcal{E}} \hat{\psi}(\epsilon) \leq \inf_{(\epsilon_0^*,\epsilon_1) \in  \mathcal{E}_0^{'} \times \mathcal{E}_1} \hat{\psi} \circ g_p(\epsilon_0^*, \epsilon_1) \leq  \sup_{(\epsilon_0^*,\epsilon_1) \in  \mathcal{E}_0^{'} \times \mathcal{E}_1} \hat{\psi} \circ g_p(\epsilon_0^*, \epsilon_1) \leq \sup_{\epsilon \in \mathcal{E}} \hat{\psi}(\epsilon)$.
	
	(ii).$W_r \leq W_g$, where $W_{g}$ and $W_{r}$ as the bound widths of the sensitivity analysis method and the RCT-enhanced method, respectively.
\end{lemma}

To further analyze the synthesis method, we need to find the relationship between the RCT-enhanced sensitivity analysis estimator and the synthesis sensitivity analysis estimator, which is in the following theorem.

\begin{theorem}\label{thm:2}
	Suppose the sensitivity parameter $\epsilon$ for the super-population satisfies Definition \ref{def:1}. For the RCT-enhanced sensitivity analysis estimator $\hat{\psi} \circ g_p(\epsilon_0^*, \epsilon_1; \theta)$ and the synthesis sensitivity analysis estimator $\hat{\psi}_{syn}(\epsilon_1) = p_0^N\hat{\psi}_{gen} + p_1^N \hat{\psi}_1(\epsilon_1)$, we have the following conclusions:
	
	(i).If $\epsilon$ is a function of covariates (i.e. $\epsilon = \epsilon(X)$) and $\hat{\psi}(\epsilon) = n^{-1}\sum_{i = 1}^{n}\{f \circ \epsilon(X)\}$, where $f$ is a known function. Then for certain $\epsilon_0^*$ and $\epsilon_1$, the two estimators are asymptotically equivalent, i.e. $\hat{\psi} \circ g_p(\epsilon_0^*, \epsilon_1; \theta) = \hat{\psi}_{syn} + o_p(1)$.
	
	(ii).For a sensitivity analysis method with only OS data, in which $\epsilon$ is the sensitivity parameter, the equation proposed in (i) does not necessarily hold. However, for a given pair of parameters $(\epsilon_0^*, \epsilon_1)$ on the left-hand side, it is always feasible to find a $\epsilon_1^{'}$ on the right-hand side to make the equation in (i) hold. For convenience, we denote the set of $\epsilon_1^{'}$ satisfying the equation above as $\mathcal{E}_{sol}$.
\end{theorem}

Theorem \ref{thm:2} serves as a bridge for comparing our synthesis method with the sensitivity analysis method with only OS data. According to the theorem, we can claim that if we choose the range of $\epsilon_1^{'}$ as a subset (or equal) to $\mathcal{E}_{sol}$, we can ensure that the synthesis estimator will have a tighter bound than the sensitivity analysis estimator with only OS data. The following corollary formally states the comparison results.

\begin{corollary}\label{cor:1}
	For a sensitivity analysis method $\hat{\psi}(\epsilon)$ with $\epsilon$ as the sensitivity parameter for the super-population satisfying Definition \ref{def:1}, and the corresponding synthesis method $\hat{\psi}_{syn}(\epsilon_1)$, where $\epsilon \in \mathcal{E}$ and $\epsilon_1 \in \mathcal{E}^{sol}$ as proposed in Theorem \ref{thm:2}, denote $W_{s}$ and $W_{g}$ as the bound widths estimated by the two methods. We have the following conclusions:
	
	(i). $W_s \leq W_g$, for a given $p_0 \in (0, 1)$.
	
	(ii). $\lim_{p_0 \to 0} W_s/W_g = 1$, and $\lim_{p_0 \to 1} W_s/W_g = 0$.
\end{corollary} 

Although we have proven the existence of $\mathcal{E}_{sol}$, a set of sensitivity parameters for the synthesis method that results in tighter bounds than the sensitivity analysis method using only OS data, finding this set in practice is infeasible for two reasons. Firstly, the transformation function in Definition \ref{def:1} contains unknown nuisance parameters, making it difficult to identify the range for $\epsilon_0$ and $\epsilon_1$ from a given range of $\epsilon$. Secondly, in the RCT-enhanced sensitivity analysis method, solving for the range of $\epsilon_0^{*}$ using the restriction in the RCT-enhanced estimator is complex. Specifically, under some trivial cases as in Theorem \ref{thm:2} (i),  $\mathcal{E}_{sol}$ can be easily identified. Example \ref{ex:1}, corresponding to the special case in Theorem \ref{thm:2} (i), illustrates this process.

Although $\mathcal{E}_{sol}$ cannot be identified in general, Corollary \ref{cor:1} helps determine if a certain range of $\epsilon_1^{'}$ is appropriate. The procedure involves (1). using the sensitivity analysis method on the OS data only to calculate a bound width $W_0$ for a specified range of $\epsilon$, and 
(2). setting a range $\mathcal{E}_{test}$ for $\epsilon_1^{'}$, and using the synthesis method to calculate the ATE bound, whose width is denoted as $W_{test}$. If $W_{test} > W_0$, conclude that $\mathcal{E}_{sol} \subset \mathcal{E}_{test}$ and shrink $\mathcal{E}_{test}$ until $W_{test} \leq W_0$. Although we may not precisely identify $\mathcal{E}_{sol}$, achieving a tighter bound is sufficient.

An example of shrinking $\mathcal{E}_{test}$ can be found in \citet{cinelli2020making}, where the confounding strength of benchmark variables is used to bound the confounding strength of the unmeasured confounder. The sensitivity parameters are bounded by  $R^2_{A \sim U\mid X} = k_A \{R^2_{A \sim X_j\mid X_{-j}}/(1 - {R^2_{A \sim X_j\mid X_{-j}}})\}$, and $R^2_{Y \sim U\mid A,X} \leq \eta^2 \{R^2_{Y \sim X_j\mid A, X_{-j}}(1 - {R^2_{Y \sim X_j\mid A, X_{-j}}})\}$
, where $\eta$ is a parameter including $k_Y$. $k_A$ and $k_Y$ represent the relative confounding strength of the unmeasured confounder compared to a benchmark variable. The following example applies this benchmark variable bounding technique to test the refinement of the synthesis method.

\begin{example}{(Example \ref{ex:2}, continued)}\label{ex:5}
	Assume the sensitivity parameter pair  $( R^2_{A \sim U\mid X}, R^2_{Y \sim U\mid A, X}) \in [0, a] \times [0, b]$, where $ 0 <a<1$ and $ 0<b \leq 1$. From the bias formula in Example \ref{ex:2}, the bound width is  $2 |\widehat{bias}(a,b)|$, where $\widehat{bias}(\cdot, \cdot)$ is the bias function defined in Example \ref{ex:2}. The following procedure graph in Table \ref{tab:3} describes the comparison.
\begin{table}[h]
	\centering
	\begin{threeparttable}
		\caption{Procedure of adjusting the range of sensitivity parameters in the synthesis method based on example \ref{ex:2}}
		\label{tab:3}
		\begin{tabular}{ll}
			\textbf{Set}& $k_A$ and $k_Y$ (usually equal to 1) \\
			\textbf{Do} & Choose a benchmark variable $X_j$: \\
			& \textit{On the entire OS data:} 
			$a =  k_Af^2_{R^2_{A \sim X_j\mid X_{-j}}}, b =  k_Yf^2_{R^2_{Y \sim X_j\mid A, X_{-j}}}, W_1 = 2 |\widehat{bias}(a,b)|$  \\
			& \textit{On the non-overlapping data:} 
			$a_1 =  k_Af^2_{R^2_{A \sim X_j\mid X_{-j}}}, b_1 =  k_Yf^2_{R^2_{Y \sim X_j\mid A, X_{-j}}},W_2 = 2 p_1 |\widehat{bias}(a_1,b_1)|$  \\
			& $ratio = \frac{W_2}{W_1}$ \\
			\textbf{If}  & $ratio \leq 1:$ \\
			& \quad Synthesis method is better, use $W_2$ \\
			\textbf{Else} & \quad Find additional benchmark variables, or use $W_1$ \\
		\end{tabular}
		\begin{tablenotes}
			\item[]$f^2_{R^2} = \frac{R^2}{1-R^2}$.
		\end{tablenotes}
	\end{threeparttable}
\end{table}
\end{example}

\section{Simulations and Application}
\subsection{Simulation Setting}
We design two experiments to compare the sensitivity analysis bounds between the sensitivity analysis method using only OS data and the synthesis sensitivity analysis method. Both the RCT and OS samples are assumed to be drawn from a super-population of size $N = 5 \times 10^5$.

\textbf{Scenario I}: We assume the covariate is a 5-dimensional vector $X = (X_1, X_2, X_3, X_4, X_5)$, where $X_j \sim \mathcal{N}(10, 4^2)$, and $X_5$
is influenced by inclusion/exclusion criteria. The OS sample is randomly sampled from the super-population, while the RCT membership variable $S$ follows the distribution:
\begin{equation*}
	pr(S = 1\mid X) = 
	\begin{cases}
		exp(-3 + 0.5X_1 -0.3 X_2 - 0.5X_3 - 0.4X_4 + 0.1 X_5), & \text{ if } \Phi(X_5) \leq q \\
		0. &\text{ otherwise } 
	\end{cases}
\end{equation*}
where $\Phi(\cdot)$ is the cumulative distribution function of $X_5$ and $q$ is the threshold, representing the RCT-overlapping proportion.

In the OS data, we assume an unmeasured confounder $U \sim \mathcal{N}(10, 3^2)$. The treatment assignment mechanism is $pr(A = 1\mid X,U) = exp(0.5X_1 - 0.5X_2 - 0.3 X_3 +0.5X_4 - 0.3X_5 + 0.2 U)$. In contrast, $A \sim Ber(0.5)$ in the RCT data due to perfect randomization. In both datasets, the average treatment effect is $\tau = 4E(X_5)$, but the outcome differs.  In RCT, $Y = 1 + A\tau + 1.5 X_2 + 2X_3+ 2X_4 + 0.5 X_5 + \mathcal{N}(0, 3^2)$, while in the OS data, $Y = 1 +  A\tau + 1.5 X_2 + 2X_3+ 2X_4 + 0.5 X_5+ 3 U +  \mathcal{N}(0, 3^2)$.

\textbf{Scenario II}: We assume a 3-dimensional covariate vector 
$X = (X_1, X_2, X_3)$, where $X_j \sim \mathcal{N}(10,2^2)$. Here, 
$X_2$ and $X_3$ are influenced by RCT inclusion/exclusion restrictions, such that the RCT membership variable $S$ follows the distribution:
\begin{equation*}
	pr(S = 1\mid X) = 
	\begin{cases}
		exp(-2 + 0.4X_1 -0.3 X_2 - 0.5X_3), & \text{ if } \Phi(X_2) \leq q \quad \text{and} \quad \Phi(X_3) \leq q\\
		0.&\text{otherwise}
	\end{cases}
\end{equation*}

In Scenario II, $\Phi(\cdot)$ and $q$ are defined similarly as in Scenario I. In the RCT data, $A \sim Ber(0.5)$, and the outcome model is $Y = 1 + 10A+ X_1 + 1.5 X_2 + 2X_3 + \mathcal{N}(0, 1)$. In the OS data, the treatment assignment mechanism is  $pr(A = 1\mid X) = exp(0.5X_1 - 0.6X_2 + 0.3 X_3 +0.5U)$ and the outcome model is $Y = 1 + 10A+ X_1 + 1.5 X_2 + 2X_3 + 0.5 U + \mathcal{N}(0, 1)$, where $U \sim \mathcal{N}(0, 4^2)$.

We use the following comparative simulations to demonstrate the advantage of our synthesis method (denoted as 'OS+RCT') over the  sensitivity analysis method using only OS data(denoted as 'OS') in narrowing the bound width. The simulations focus on the following two cases in each scenario:(i). with a fixed $q = 0.7$, compare the bound width between the two methods. (ii). range $q \in (0, 1)$, calculate the ratio of bound width (bound width of synthesis method over sensitivity analysis method)

In Scenario I, we use the methodology in Example \ref{ex:1}, and the methodology in Example \ref{ex:2} is applied in Scenario II. We use the bootstrap method to calculate the standard deviation of the upper and lower bounds.

\subsection{Results}

\textbf{Scenario I}: For a fixed $q = 0.7$, we set the sensitivity parameters $(\epsilon^0, \epsilon^1) \in [0.9, 1.1] \times [0.9,1.1]$, and the corresponding  $(\epsilon^0_1, \epsilon^1_0) \in [0.9, 1.1] \times [0.9,1.1]$. Our sensitivity analysis methods in Example \ref{ex:1} include the modified outcome regression estimator, Horvitz–Thompson-type inverse propensity score weighting estimator, Hajek-type estimator, and double robust estimator. In the synthesis method, we use the outcome-model-based estimator, inverse probability of sampling estimator , and augmented inverse probability of sampling estimator to generalize RCT results to the RCT-eligible real-world population.

The comparison results, including the mean and standard deviation (in the brackets) of the lower and upper bounds and the mean of bound width, are presented in Table \ref{tab:4}. The synthesis estimator combining RCT data narrows the sensitivity analysis bound width in all cases. However, when the RCT generalization method is IPSW or the sensitivity analysis method is HT, the synthesis method does not always cover the true value with the same sensitivity parameter range. This is because these methodologies are unstable, and extreme weights may occur, resulting in a large standard deviation.

\begin{table}[h]
	\centering
	
	\begin{threeparttable}
		\caption{Sensitivity analysis bounds comparison between synthesis method and original method based on modified estimators in Example \ref{ex:1}}
		\label{tab:4}
			\begin{tabular}{cccccc}
			&	& OS & OS+RCT(OM) & OS+RCT(IPSW) & OS+RCT(AIPSW) \\
			& lb & 33.65(0.76) &36.92(1.12) &28.23(9.01) &36.94(1.09) \\
			OR  & ub &57.44(0.68) &45.23(1.19) &36.58(9.06) &45.30(1.16) \\
			& mbw & 23.79&8.32 &8.35 & 8.36\\
			& lb &35.23(6.78) &31.83(4.27) & 22.98(10.10)& 31.86(4.38)\\
			HT  & ub & 54.25(6.35) &38.03(4.13) &29.16(9.99) & 38.05(4.17)\\
			& mbw & 19.01&6.20 &6.18 &6.19 \\
			& lb & 32.49(1.30)&36.60(1.17) &29.49(9.33) & 36.51(1.12)\\
			Haj  & ub &53.80(1.16) &43.63(1.17) &36.56(7.07) &43.59(1.22) \\
			& mbw &21.31 &7.02 & 7.08& 7.08\\
			& lb &33.94(0.98) &37.14(1.20) &28.54(9.13) & 36.97(1.05)\\
			DR  & ub & 57.93(0.94)&45.51(1.26) &36.92(9.14) & 45.39(1.17)\\
			& mbw &23.99 &8.37 &8.37 & 8.42\\
		\end{tabular}
	\end{threeparttable}
	\begin{tablenotes}
		\item OR: Outcome regression estimator; HT, Horvitz–Thompson-type inverse propensity score weighting estimator; Haj, Hajek-type inverse propensity score weighting estimator; DR, double robust estimator; lb, lower bound; ub, upper bound; mbw, mean bound width. Other abbreviations refer to Table \ref{tab:1}.
	\end{tablenotes}
\end{table}

When ranging $q$ in $[0,1]$, the bound width ratio also changes. The trend in Figure \ref{fig:1} aligns with our intuition. When $q \to 0$, it indicates that no RCT data can be integrated into the synthesis method, making the two sensitivity analysis methods equivalent. Consequently, the ratio tends to 1. Conversely, when $q \to 1$, the ATE can be consistently estimated from RCT data using the generalization method, eliminating the need for sensitivity analysis. Therefore, the ratio is 0.

\begin{figure}[h]
	\centering
	\includegraphics[width=0.6\linewidth]{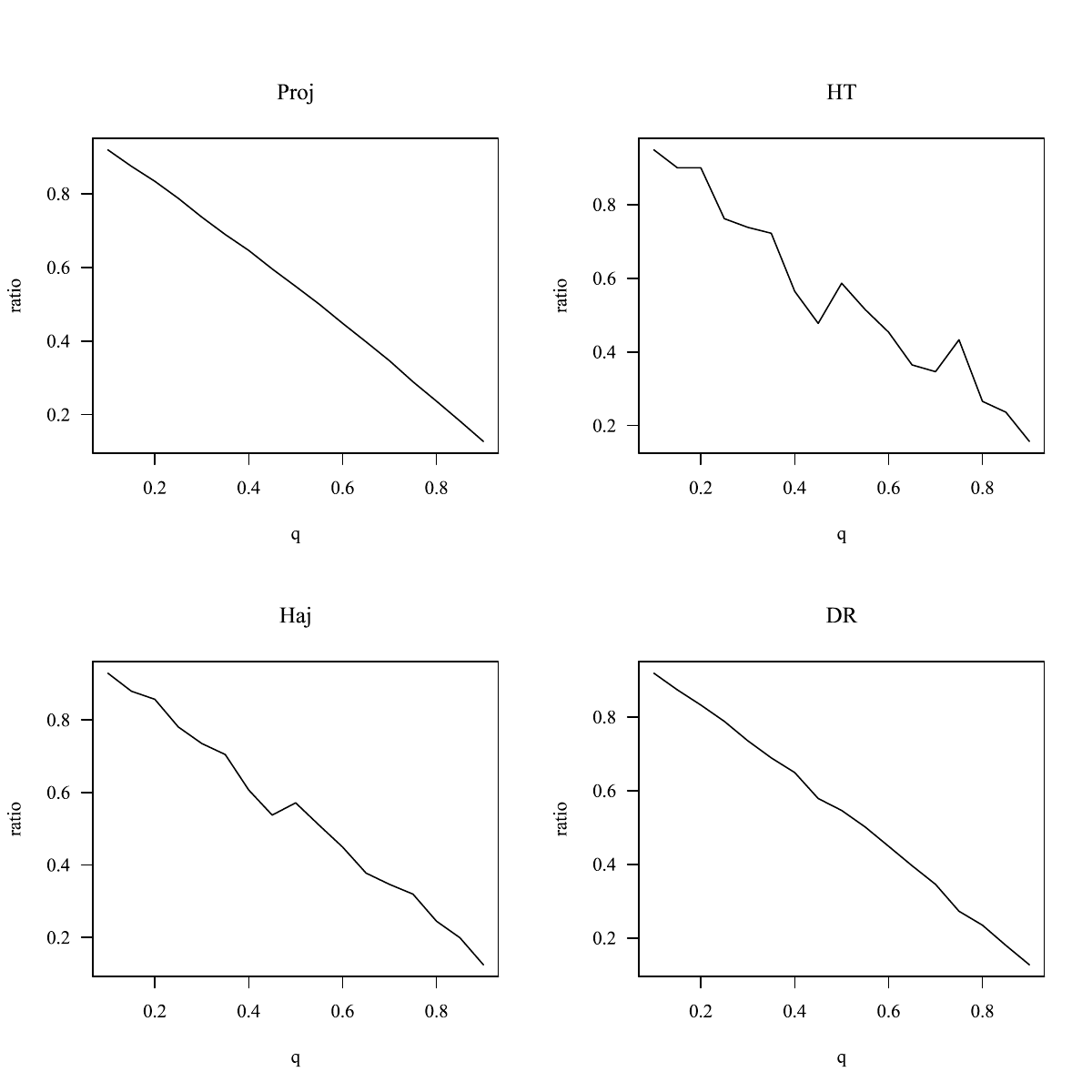}
	\caption{Sensitivity analysis bound width ratio with a varing RCT-overlap proportion.}
\label{fig:1}
\end{figure}

\textbf{Scenario II}: In this scenario, with a fixed $q = 0.7$, we apply the synthesis of sensitivity analysis methods based on the omitted variable bias framework introduced in Example \ref{ex:2} and RCT generalization methods to estimate the upper and lower bounds, as well as the mean bound width. Here, we use $X_1$
as the benchmarking variable to bound the confounding strength of the unmeasured confounder, following the procedure in Example \ref{ex:5}. We set $k_D = 6$ and $k_Y = 1$ to bound the confounding strength. Additionally, we consider fixed value bounds for the confounding strength with $R^2_{A \sim U\mid X,V^* = j} \leq 0.1$ and $R^2_{A \sim U\mid X,V^* = j} \leq 0.8$. For varying $q \in [0,1]$, we set  $R^2_{A \sim U\mid X,V^* = j} \leq 0.1$ and $R^2_{A \sim U\mid X,V^* = j} \leq 0.9$, then calculate the theoretical parameter bound on the entire OS data using the transformation function proposed in Definition \ref{def:1}. Therefore the bound width ratio of the synthesis method over the general method is calculated. All results containing a table of bound width for a fixed $q$ and a plot of bound width ratio for a varying $q$ are relegated into the supplemental material due to length limitation of the paper, which aligns with our intuition.

\subsection{Application: Comparison of Effectiveness Between Two BP-lowering Drugs}

We perform a sensitivity analysis to estimate the efficacy difference in lowering blood pressure between two antihypertension drugs: Songling Xuemaikang Capsule (SXC), a traditional Chinese medicine, and Losartan, a widely accepted drug. A comparative-effectiveness RCT with 602 participants assessed their efficacy by measuring reductions in systolic (SBP) and diastolic blood pressure (DBP). The RCT found that SXC was non-inferior to Losartan in patients with grade I hypertension \citep{lai2022efficacy, lai2023propensity}. However, the RCT included only patients with grade I hypertension aged 60 or younger.In contrast, the OS data was collected from 3000 hypertensive patients without these restrictions. In the OS data, patients taking SXC were compared to those on other antihypertensive medications. Both datasets recorded baseline SBP and DBP (BSBP and BDBP), as well as follow-up measures at weeks 2, 4, 6, and 8. Efficacy was defined as the difference between follow-up and baseline measures (DSBP and DDBP). Since the control group in the OS used Losartan or similar medications, we assume their outcomes reflect what would have been observed if they had used Losartan.

Among the 3000 patients in the OS data, 1160 were excluded based on RCT inclusion/exclusion criteria. We aim to conduct a sensitivity analysis on the efficacy comparison at week 8 for the entire OS population. This analysis utilizes the omitted variable bias framework from Example \ref{ex:2}. To integrate the RCT data, we use the augmented inverse probability of sampling weighting estimator (AIPSW estimator) to generalize the RCT results to the overlapping real-world population.

We compare SBP reduction efficacy using the benchmark bounding technique from \citet{cinelli2020making} to evaluated unmeasured confounder strength. According to covariates benchmarking bounding results in the supplementary material, we use BSBP as the bound for the unmeasured confounding strength conservatively.  Results for the  sensitivity analysis method using only the OS data and our synthesis method are presented in Figure \ref{fig:2}, with detailed results in the supplementary material. The former method yields a lower bound of -3.55 (95\% CI: [-4.43, -2.63]), an upper bound of 4.96 (95\% CI: [4.07, 5.94]), and a mean bound width of 8.52. In comparison, our synthesis method provides a lower bound of -2.61 (95\% CI: [-3.96, -1.44]), an upper bound of 3.31 (95\% CI: [2.05, 4.47]), and a mean bound width of 5.92. Both methods indicate no statistically significant difference in SBP reduction between SXC and Losartan at week 8, but our synthesis method offers a narrower bound. As for DDBP at week 8, the comparison results imply the same conclusion, which is relegated to the supplemental material due to the limitation of the length of the paper.
\begin{figure}[h]
	\centering
	\subfigure{\includegraphics[width=.4\textwidth]{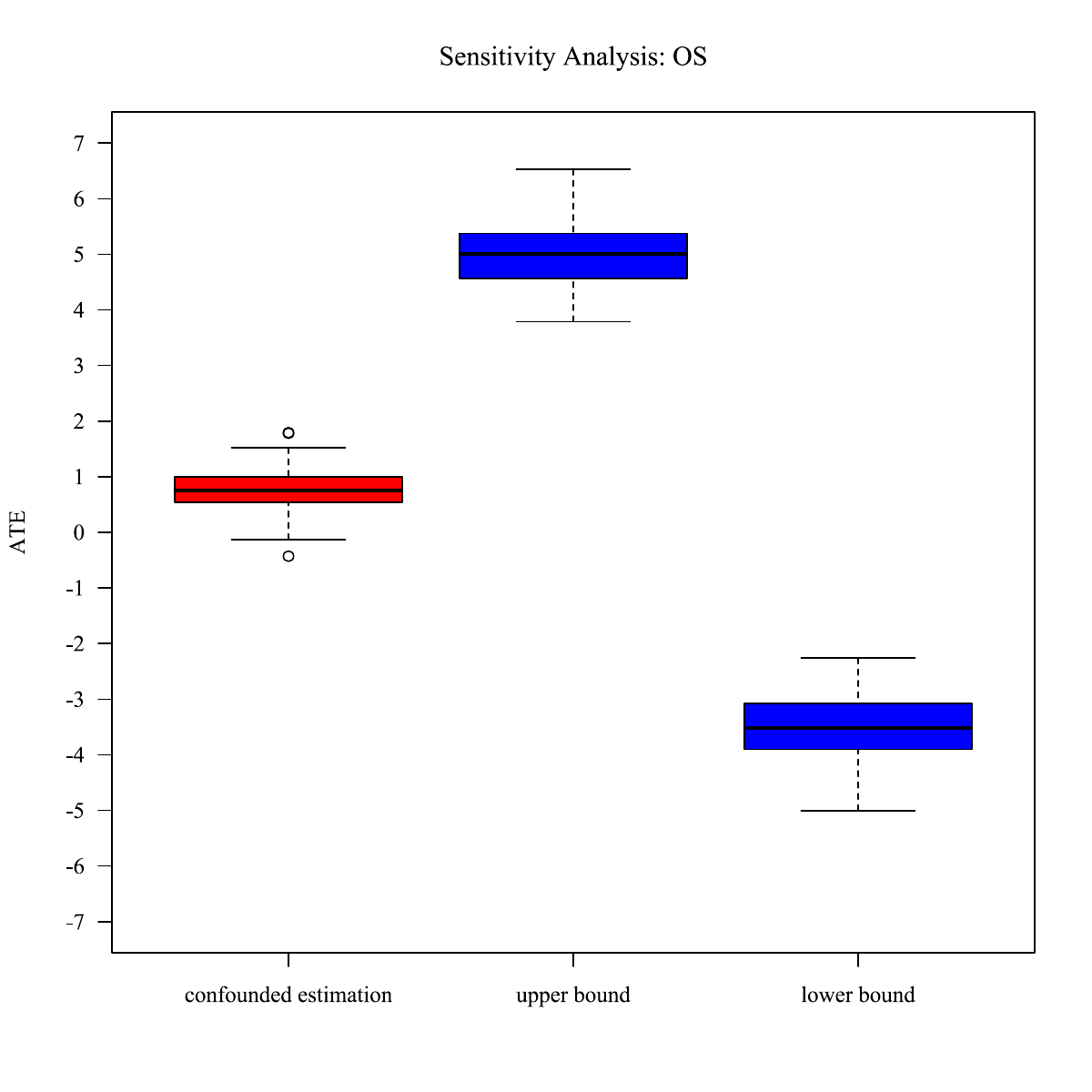}}
	\subfigure{\includegraphics[width=.4\textwidth]{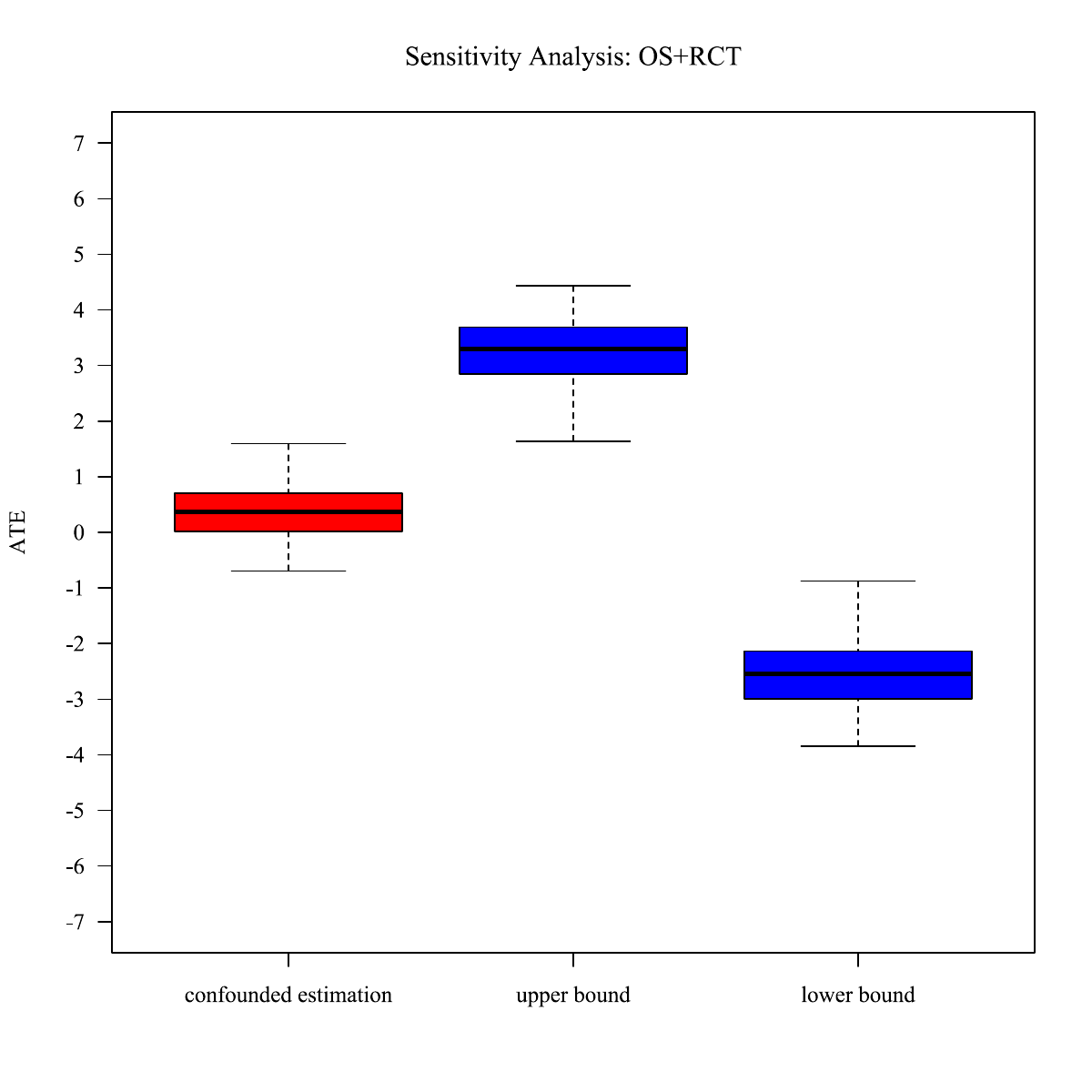}}
	\caption{Comparison of sensitivity analysis bounds between general sensitivity analysis method in Example \ref{ex:2} and synthesis method (RCT generalization method: AIPSW). Outcome: DSBP at week 8.}
	\label{fig:3}
\end{figure}

\section{Discussion}

In this paper, we integrate RCT generalization methods with sensitivity analysis techniques using both the RCT data and OS data. Our synthesis sensitivity analysis method offers two key advantages: firstly, it gives tighter bounds through integrating additional RCT data compared with the classical sensitivity analysis methods using only OS data;  besides, it addresses positivity assumption violations common in RCT generalization methods. 
When this assumption is violated, the RCT sample represents only the trial-eligible population, not the entire population \citep{dahabreh2019generalizing, lee2023improving}. Our approach gives a partial solution for the violation of this assumption. However, our focus is on explicit inclusion/exclusion criteria set during RCT design. For implicit exclusions in covariates, distinguishing between RCT-represented and unrepresented populations requires further research \citep{parikh2024missing}.

The performance of the synthesis method depends on the RCT-overlapping proportion $p_0$ with the OS data. As indicated by Equation \eqref{eq:1}, if  $p_0 = 1$, sensitivity analysis is unnecessary because the positivity assumption is satisfied and the ATE can be consistently estimated. Conversely, if $p_0 = 0$, which means that  RCT data is not available, we can  only use sensitivity analysis methods in the observational data. Additionally, the choice of RCT generalization method affects the behavior of the bounds. According to Theorem \ref{thm:1}, the variance of the bound can vary based on the chosen RCT generalization method. For instance, using the IPSW estimator as the generalization method results in larger variance, as shown in the simulations. For more details on different RCT generalization methods, refer to \citet{dahabreh2019generalizing, lee2023improving}.

There are some other limitations in this work. Firstly, 
Assumption \ref{ass:3} is critical for the consistency of $\hat{\psi}_{gen}$. If this assumption is violated, the estimator is inconsistent, complicating the sensitivity analysis and the evaluation of our synthesis method \citep{nguyen2018sensitivity,huang2022sensitivity,hartman2024sensitivity}. Secondly, in this paper we only consider sensitivity analysis methods with separable sensitivity parameters, which are only a subset of all possible sensitivity analysis methods. Thirdly, our discussion is limited to certain outcomes, excluding more complex cases like survival outcomes \citep{lee2022doublyrobustestimatorsgeneralizing,lee2024transportingsurvival}.  These topics remain for future researches.

\newpage

\bibliography{paper-ref-abb.bib}

\clearpage
\newpage
\setcounter{page}{1}

\appendix
\setcounter{lemma}{0}
\setcounter{equation}{0}
\setcounter{proposition}{0}
\setcounter{theorem}{0}
\setcounter{figure}{0}
\setcounter{table}{0}
\renewcommand{\thefigure}{B\arabic{figure}}
\renewcommand{\thetable}{B\arabic{table}}
\renewcommand{\theequation}{A\arabic{equation}}
\renewcommand{\thelemma}{A\arabic{lemma}}
\renewcommand{\thetheorem}{A\arabic{theorem}}
\renewcommand{\theproposition}{A\arabic{proposition}}

\setcounter{definition}{0}
\renewcommand{\thedefinition}{A\arabic{definition}}

\section*{Appendix}
The Supplemental Material is organized into two sections. Section \ref{sec:A} contains detailed proofs of the theorems and examples proposed in the main paper. Section \ref{sec:B} provides complementary simulations results and the real-data analysis results comparing the efficacy of SXC and Losartan, which were relegated from the main text due to length constraints.

\section{Technical Proofs}\label{sec:A}
\subsection{Proof of Theorem 1.(ii)}

For a correctly specified $\epsilon_1$,  we know that $\hat{\psi}_{syn}(\epsilon_1) \to \psi$ in probability, where $\psi = p_0\psi_0 + p_1\psi_1(\epsilon_1)$. We have:

\begin{equation*} 
	\begin{aligned}
		N^{1/2}\{\hat{\psi}_{syn}(\epsilon_1) - \psi\} &= N^{1/2} \left\{\frac{N_0}{N}\hat{\psi}_{gen} + \frac{N_1}{N}\hat{\psi}_1(\epsilon_1) - p_0 \psi_0 - p_1\psi_1(\epsilon_1)\right\}\\
		& = N^{1/2} \left\{(\frac{N_0}{N} - p_0) \hat{\psi}_{gen} + p_0(\hat{\psi}_{gen} - \psi_0)\right\}  \\
		&  \quad  + N^{1/2} \left[ (\frac{N_1}{N} - p_1)\hat{\psi}_1(\epsilon_1) + p_1\{\hat{\psi}_1(\epsilon_1) - \psi_1(\epsilon_1)\}   \right]  \\
		& = N^{1/2}(\frac{N_0}{N} - p_0)\{\hat{\psi}_{gen} - \hat{\psi}_1(\epsilon_1)\}   \\
		& \quad +  N^{1/2}/N_0^{1/2} p_0\times N_0^{1/2}(\hat{\psi}_{gen} - \psi_0) + N^{1/2}/N_1^{1/2} p_1 \times N_1^{1/2}\{\hat{\psi}_{1}(\epsilon_1) - \psi_1(\epsilon_1)\}  \\
		& \to \mathcal{N}[0, p_0(1 - p_0)\{\psi_0 - \psi_1(\epsilon_1)\}^2 + p_0\sigma_0^2 + p_1 \sigma_1^2].
	\end{aligned}
\end{equation*}
The third equation holds due to the fact that $p_0 + p_1 = 1$
and $N_0/N + N_1/N = 1$.

where $\sigma_0^2$ and $\sigma_1^2$ are the asymptotic variance of $\hat{\psi}_{gen}$ and $\hat{\psi}_1(\epsilon_1)$, respectively. The third equation holds due to the fact that $p_0 + p_1 = 1$
and $N_0/N + N_1/N = 1$.

\subsection{Proof of Example 4}
This subsection proves the separability of sensitivity analysis parameters of the omitted variable bias framework in \citet{cinelli2020making,chernozhukov2022long}. For a sub-dataset with $V^* = j$, $j \in \{0,1\}$, the sensitivity parameters on are:
\begin{equation}\label{eq:A1}
	R^2_{Y \sim U\mid A, X,V^* = j }= \frac{var\{E(Y\mid U,A,X,V^* = j)\mid V^* = j\} - var\{E(Y\mid A,X,V^* = j)\mid V^* = j\}}{var(Y \mid V^* = j) - var\{E(Y\mid A,X,V^* = j)\mid V^* = j\}}, 
\end{equation}

\begin{equation}\label{eq:A2}
	R^2_{A \sim U\mid X,V^* = j} = \frac{var\{E(A\mid U,X,V^* = j)\mid V^* = j\} - var\{E(A\mid X,V^* = j)\mid V^* = j\}}{var(Y \mid V^* = j ) - var\{E(A\mid X,V^* = j)\mid V^* = j\}}.
\end{equation}

For the proof, we need the following proposition firstly.

\begin{proposition} \label{pro:A1}
	For $j \in \{0,1\}$, the following equations hold:
	\begin{equation*}
		var\{E(Y\mid U,A,X)\mid V^* = j\} =  var\{E(Y\mid U,A , X , V^* = j) \mid V^* = j\},  
	\end{equation*}
	and 
	\begin{equation*}
		var\{E(Y\mid A,X)\mid V^* = j\} = var\{E(Y\mid A , X, V^* = j) \mid V^* = j\}.
	\end{equation*}
\end{proposition}

\begin{proof}

	For the first equation, we have:
	\begin{equation*}
		\begin{aligned}
			E(Y\mid U = u,A = a,X = x] &= \int y pr(y\mid u,a,x) dy \\
			& = \int \sum_{j = 0}^1 y pr(y\mid u,a,x,V^*= j) pr(V^* = j\mid U =  u, A =  a, X = x) dy \\
			& = \sum_{j = 0}^1\left\{ \int y pr(y\mid u,a,x,V^*= j) dy \right\} pr(V^* = j\mid U =  u, A =  a, X =  x) \\
			& = \sum_{j = 0}^1 E(Y\mid U = u,A = a, X = x, V^* = j) pr(V^* = j\mid U =  u, A = a, X =  x)   \\
			& = \sum_{j = 0}^1 E(Y\mid U = u,A = a, X = x, V^* = j) pr(V^* = j \mid X = x)  \\
			& = \sum_{j = 0}^1 E(Y\mid U = u,A = a, X = x, V^* = j) I(V^* = j),\\
		\end{aligned}
	\end{equation*}
	where the fifth and the sixth equation hold because the inclusion/exclusion variable $V^*$ solely depends on $X$.
	
	Therefore,
	\begin{equation*}
		\begin{aligned}
			var\{E(Y\mid U,A,X)\mid V^* = j\} & = var \left\{\sum_{k = 0}^J E(Y\mid U ,A , X, V^* = j) I(V^* = k)\mid V^* = j \right\} \\
			& = var\{E(Y\mid U,A , X , V^* = j) \mid V^* = j\}. \\ 
		\end{aligned}
	\end{equation*}
	
	Similarly, we have:
	\begin{equation*}
		var\{E(Y\mid A,X)\mid V^* = j\} = var\{E(Y\mid A , X, V^* = j) \mid V^* = j\}.
	\end{equation*}
	
\end{proof}

Then we will separate the parameter $R^2_{Y \sim U\mid A, X}$ based on the proposition above. Following the law of total variance, we have:
\begin{equation*}
	\begin{aligned}
		R^2_{Y \sim U\mid A, X} & = \frac{var\{E(Y\mid U,A,X)\} - var\{E(Y\mid A,X)\}}{var(Y) - var\{E(Y\mid A,X)\}}  \\
		& = \frac{p_0 var\{E(Y\mid U,A,X)\mid V^* = 0\} + p_1var\{E(Y\mid U,A,X)\mid V^* = 1\}+ var[E\{E(Y\mid U,A,X)\mid V^*\}]} {p_0var(Y\mid V^* = 0) + p_1var(Y\mid V^* = 1)+ var\{E(Y\mid V^*)\}} \\
		& \quad \qquad \frac{- (p_0var\{E(Y\mid A,X)\mid V^* = 0\} +p_1 var\{E(Y\mid A,X)\mid V^* = 1\} + var[E\{E(Y\mid A,X)\mid V^*\}])}{ - (p_0var\{E(Y\mid A,X)\mid V^* = 0\} + p_1var\{E(Y\mid A,X)\mid V^* = 1\} + var[E\{E(Y\mid A,X)\mid V^*\}])}  \\
		& = \frac{p_0var\{E(Y\mid U,A,X,V^* = 0)\mid V^* = 0\} + p_1var\{E(Y\mid U,A,X,V^* = 1)\mid V^* = 1\}+ var[E\{E(Y\mid U,A,X)\mid V^*\}]} {p_0var(Y\mid V^* = 0) + p_1var(Y\mid V^* = 1)+ var\{E(Y\mid V^*)\}} \\
		& \quad \qquad \frac{- (p_0var\{E(Y\mid A,X,V^* = 0)\mid V^* = 0\} +p_1 var\{E(Y\mid A,X,V^* = 1)\mid V^* = 1\} + var[E\{E(Y\mid A,X)\mid V^*\}])}{ - (p_0var\{E(Y\mid A,X,V^* = 0)\mid V^* = 0\} + p_1var\{E(Y\mid A,X,V^* = 1)\mid V^* = 1\} + var[E\{E(Y\mid A,X)\mid V^*\}])}  \\
		& = \frac{p_0[var\{E(Y\mid U,A,X,V^* = 0)\mid V^* = 0\} - var\{E(Y\mid A,X,V^* = 0)\mid V^* = 0\}]}{p_0[var(Y\mid V^* = 0) - var\{E(Y\mid A,X,V^* = 0)\mid V^* = 0\}]} \\
		& \quad \qquad \frac{+ p_1[var\{E(Y\mid U,A,X,V^* = 1)\mid V^* = 1\} - var\{E(Y\mid A,X,V^* = 1)\mid V^* = 1\}]}{+ p_1[var(Y\mid V^* = 1) - var\{E(Y\mid A,X,V^* = 1)\mid V^* = 1\}]}  \\
		& \quad \qquad \quad \qquad \frac{+  var[E\{E(Y\mid U,A,X)\mid V^*\}] - var[E\{E(Y\mid A,X)\mid V^*\}]}{ + var[E\{E(Y\mid A,X)\mid V^*\}] - var[E\{E(Y\mid A,X)\mid V^*\}] } \\  
		& = \frac{p_0 \frac{var\{E(Y\mid U,A,X,V^* = 0)\mid V^* = 0\} - var\{E(Y\mid A,X,V^* = 0)\mid V^* = 0\}}{[var(Y\mid V^* = 0) - var\{E(Y\mid A,X,V^* = 0)\mid V^* = 0\}][var(Y\mid V^* = 1) - var\{E(Y\mid A,X,V^* = 1)\mid V^* = 1\}]}}{p_0\frac{1}{var(Y\mid V^* = 1) - var\{E(Y\mid A,X,V^* = 1)\mid V^* = 1\}}} \\
		& \quad \qquad \frac{+ p_1\frac{var\{E(Y\mid U,A,X,V^* = 1)\mid V^* = 1\} - var\{E(Y\mid A,X,V^* = 1)\mid V^* = 1\}}{[var(Y\mid V^* = 0) - var\{E(Y\mid A,X,V^* = 0)\mid V^* = 0\}][var(Y\mid V^* = 1) - var\{E(Y\mid A,X,V^* = 1)\mid V^* = 1\}]}}{+ p_1\frac{1}{var(Y\mid V^* = 0) - var\{E(Y\mid A,X,V^* = 0)\mid V^* = 0\}}}  \\
		& \quad \qquad \quad \qquad \frac{+ \frac{var[E\{E(Y\mid U,A,X)\mid V^*\}] - var[E\{E(Y\mid A,X)\mid V^*\}]}{[var(Y\mid V^* = 0) - var\{E(Y\mid A,X,V^* = 0)\mid V^* = 0\}][var(Y\mid V^* = 1) - var\{E(Y\mid A,X,V^* = 1)\mid V^* = 1\}]}}{ + \frac{var[E\{E(Y\mid A,X)\mid V^*\}] - var[E\{E(Y\mid A,X)\mid V^*\}]}{[var(Y\mid V^* = 0) - var\{E(Y\mid A,X,V^* = 0)\mid V^* = 0\}][var(Y\mid V^* = 1) - var\{E(Y\mid A,X,V^* = 1)\mid V^* = 1\}]} } \\
		& = \frac{p_0\frac{R^2_{Y \sim U\mid A, X, V^* = 0}}{var(Y\mid V^* = 1) -  var\{E(Y\mid A,X) \mid V^* = 1\}} + p_1\frac{R^2_{Y \sim U\mid A, X, V^* = 1}}{var(Y\mid V^* = 0) -  var\{E(Y\mid A,X)\mid V^* = 0\}}}{p_0\frac{1}{var(Y\mid V^* = 1) -  var\{E(Y\mid A,X)\mid V^* = 1\}} + p_1\frac{1}{var(Y\mid V^* = 0) -  var\{E(Y\mid A,X)\mid V^* = 0\}}} \\
		& \quad \qquad \frac{+ \frac{ var[E\{E(Y\mid U,A,X)\mid V^*\}] - var[E\{E(Y\mid A,X)\mid V^*\}]}{[var(Y\mid V^* = 1) -  var\{E(Y\mid A,X)\mid V^* = 1\}][var(Y\mid V^* = 0) -  var\{E(Y\mid A,X)\mid V^* = 0\}]}}{+\frac{var\{E(Y\mid V^*)\} - var[E\{E(Y\mid A,X)\mid V^*\}]}{[var(Y\mid V^* = 1) -  var\{E(Y\mid A,X)\mid V^* = 1\}][var(Y\mid V^* = 0) -  var\{E(Y\mid A,X)\mid V^* = 0\}]}}.
	\end{aligned}
\end{equation*}

The separability of the sensitivity parameter $R^2_{A \sim U\mid X}$ can be shown by the same procedure.

\begin{equation*}
	\begin{aligned}
		R^2_{A \sim U\mid X} & = \frac{var\{E(A\mid A,X)\} - var\{E(A\mid X)\} }{var(A) - var\{E(A\mid X)\}}  \\
		& = \frac{p_0 var\{E(A\mid U,X)\mid V^* = 0\} + p_1var\{E(A\mid U,X)\mid V^* = 1\}+ var[E\{E(A\mid U,X)\mid V^*\}]} {p_0 var(A\mid V^* = 0) + p_1 var(A\mid V^* = 1)+ var\{E(A\mid V^*)\}} \\
		& \quad \qquad \frac{- (p_0var\{E(A\mid X)\mid V^* = 0\} +p_1 var\{E(A\mid X)\mid V^* = 1\} + var[E\{E(A\mid X)\mid V^*\}])}{ - (p_0var(E(A\mid X)\mid V^* = 0) + p_1var(E(A\mid X)\mid V^* = 1) + var[E\{E(A\mid X)\mid V^*\}])}  \\
		& = \frac{p_0var\{E(A\mid U,X,V^* = 0)\mid V^* = 0\} + p_1var\{E(A\mid U,X,V^* = 1)\mid V^* = 1\}+ var\{E(E(A\mid U,X)\mid V^*)\}} {p_0var(A\mid V^* = 0) + p_1var(A\mid V^* = 1)+ var\{E(A\mid V^*)\}} \\
		& \quad \qquad \frac{- (p_0var\{E(A\mid X, V^* = 0)\mid V^* = 0\} +p_1 var\{E(A\mid X,V^* = 1)\mid V^* = 1\} + var[E\{E(A\mid X)\mid V^*\}])}{ - (p_0var\{E(A\mid X,V^* = 0)\mid V^* = 0\} + p_1var(E(A\mid X,V^* = 1)\mid V^* = 1) + var[E\{E(A\mid X)\mid V^*\}])}  \\
		& = \frac{p_0 [var\{E(A\mid U,X,V^* = 0)\mid V^* = 0\} - var\{E(A\mid X, V^* = 0)\mid V^* = 0\}]}{p_0[var(A\mid V^* = 0) - var\{E(A\mid X,V^* = 0)\}]} \\
		& \quad \qquad \frac{+p_1[var\{E(A\mid U,X,V^* = 1)\mid V^* = 1\} - var\{E(A\mid X,V^* = 1)\mid V^* = 1\}]}{+p_1[var(A\mid V^* = 1) - var\{E(A\mid X,V^* = 1)\}} \\
		& \quad \qquad \quad \qquad \frac{+var\{E(E(A\mid U,X)\mid V^*)\} - var[E\{E(A\mid X)\mid V^*\}]}{+ var\{E(A\mid V^*) - var[E\{E(A\mid X)\mid V^*\}]} \\ %
		& = \frac{p_0 \frac{var\{E(A\mid U,X,V^* = 0)\mid V^* = 0\} - var\{E(A\mid X, V^* = 0)\mid V^* = 0\}]}{[var(A\mid V^* = 0) - var\{E(A\mid X,V^* = 0)\}][var(A\mid V^* = 1) - var\{E(A\mid X,V^* = 1)\}]}}{p_0\frac{1}{var(A\mid V^* = 1) - var\{E(A\mid X,V^* = 1)\}}}\\%
		& \quad \qquad \frac{ + p_1 \frac{var\{E(A\mid U,X,V^* = 1)\mid V^* = 1\} - var\{E(A\mid X,V^* = 1)\mid V^* = 1\}}{[var(A\mid V^* = 0) - var\{E(A\mid X,V^* = 0)\}][var(A\mid V^* = 1) - var\{E(A\mid X,V^* = 1)\}]}}{ + p_1 \frac{1}{var(A\mid V^* = 0) - var\{E(A\mid X,V^* = 0)\}}} \\
		& \quad \qquad \quad \qquad \frac{+ \frac{var\{E(E(A\mid U,X)\mid V^*)\} - var[E\{E(A\mid X)\mid V^*\}]}{[var(A\mid V^* = 0) - var\{E(A\mid X,V^* = 0)\}][var(A\mid V^* = 1) - var\{E(A\mid X,V^* = 1)\}]}}{+ \frac{var\{E(A\mid V^*) - var[E\{E(A\mid X)\mid V^*\}]}{[var(A\mid V^* = 0) - var\{E(A\mid X,V^* = 0)\}][var(A\mid V^* = 1) - var\{E(A\mid X,V^* = 1)\}]}} \\
		& = \frac{p_0\frac{R^2_{A \sim U\mid X,V^* = 0} }{var(A\mid V^* = 1) -  var\{E(A\mid X)\mid V^* = 1\}} + p_1\frac{R^2_{A \sim U\mid X,V^* = 1} }{var(A\mid V^* = 0) -  var\{E(A\mid X)\mid V^* = 0\}}}{p_0\frac{1}{var(A\mid V^* = 1) -  var\{E(A\mid X)\mid V^* = 1\}} + p_1\frac{1}{var(A\mid V^* = 0) -  var(E(A\mid X)\mid V^* = 0)}} \\
		& \quad \qquad \frac{+ \frac{var[E\{E(A\mid U,X)\mid V^*\}] - var[E\{E(A\mid X)\mid V^*\}]}{[var(A\mid V^* = 1) -  var\{E(A\mid X)\mid V^* = 1\}][var(A\mid V^* = 0) -  var\{E(A\mid X)\mid V^* = 0\}]}}{+\frac{var\{E(A\mid V^*)\} - var[E\{E(A\mid X)\mid V^*\}]}{[var(A\mid V^* = 1) -  var\{E(A\mid X)\mid V^* = 1\}][var(A\mid V^* = 0) -  var\{E(A\mid X)\mid V^* = 0\}]}} .
	\end{aligned}
\end{equation*}

The equations above include two unknown nuisance parameters, denoted in our paper as $\theta = (var[E\{E(Y\mid U,A,X)\mid V^*\}], var[E\{E(A\mid U,X)\mid V^*\}])$. We need to verify that these nuisance parameters are not deterministic by the sensitivity parameters in sub-datasets. Thus, the nuisance parameters can be treated as fixed values when varying the sensitivity parameters.

First, we explicitly express the nuisance parameters in the following forms:
\begin{equation*}
	\begin{aligned}
		var[E\{E(Y\mid U,A,X)\mid V^*\}] & = E[E\{E(Y \mid U,A,X)|V^*\}^2] - E[E\{E(Y \mid U,A,X) \mid V^*\}] \\
		& = p_0E\{E(Y \mid U,A,X)|V^* = 0\}^2 + p_1E\{E(Y \mid U,A,X)|V^* = 1\}^2  \\
		& \quad \qquad -[p_0 E\{E(Y \mid U,A,X)|V^* = 0\} + p_1 E\{E(Y \mid U,A,X)|V^* = 1\}], \\
		var[E\{E(A\mid U,X)\mid V^*\}] & = E[E\{E(A \mid U,X)|V^*\}^2] - E[E\{E(A \mid U,X) \mid V^*\}] \\
		& = p_0E\{E(A \mid U,X)|V^* = 0\}^2 + p_1E\{E(A \mid U,X)|V^* = 1\}^2  \\
		& \quad \qquad -[p_0 E\{E(A \mid U,X)|V^* = 0\} + p_1 E\{E(A \mid U,X)|V^* = 1\}]. \\
	\end{aligned}
\end{equation*}

The representations show that the nuisance parameters are actually functions of expectation of conditional expectations on sub-datasets. In the definition of $R^2_{Y \sim U\mid A, X,V^* = j }$ (equation \ref{eq:A1}) and $R^2_{A \sim U\mid X,V^* = j}$ (\ref{eq:A2}), and according to proposition \ref{pro:A1},  these parameters depend on  $var\{E(Y\mid U,A,X)\mid V^* = j\}$ and $var\{E(A\mid U,X)\mid V^* = j\}$ respectively, which are variances of conditional expectations. Since expectation and variance of a random variable do not determine each other, the nuisance parameters can be treated as fixed values when varying the sensitivity parameters.

\subsection{Proof of Lemma 1} \label{sec:A3}


For a given sensitivity parameter $\epsilon$, with its corresponding transformation function $g_p$, we can represent the ATE estimator as 
\begin{equation}\label{eq:A3}
	\hat{\psi} = \hat{\psi}(\epsilon) = \hat{\psi} \circ g_p(\epsilon_0, \epsilon_1; \theta), 
\end{equation}
where $\epsilon_0 \in \mathcal{E}_0, \epsilon_1 \in \mathcal{E}_1$. The equation substitutes the domain of $\epsilon$, denoted as $\mathcal{E}$, with the domain of $(\epsilon_0, \epsilon_1)$. Since the parameters are separable, the domain of the right-hand side of equation \ref{eq:A3} is $\mathcal{E}_0 \times \mathcal{E}_1$.

With a consistent estimation for $\psi_0$ as $\hat{\psi}_{gen}$ through RCT generalization methods, to find the supremum of the ATE, we need to solve the following optimization problem, named as the RCT-enhanced estimator:
\begin{equation}\label{eq:A4}
	\begin{matrix}
		\sup_{(\epsilon_0,\epsilon_1) \in \mathcal{E}} \hat{\psi} \circ g_p(\epsilon_0, \epsilon_1,\theta),\\
		\text{subject to } \hat{\psi}_0(\epsilon_0) = \hat{\psi}_{gen}.
	\end{matrix}
\end{equation}
We denote the RCT-restricted domain of $\epsilon_0$ is:
\begin{equation*}
	\mathcal{E}_0^{'} = \{\epsilon_0^*: \hat{\psi}_0(\epsilon_0^*) = \hat{\psi}_{gen}, \epsilon_0^* \in \mathcal{E}_0\}.
\end{equation*}
Obviously, $\mathcal{E}_0^{'} \subset \mathcal{E}_0$ because the latter domain is a set without any additional restrictions. Then the RCT-restricted domain of equation \ref{eq:A3} is $ \mathcal{E}_0^{'} \times \mathcal{E}_1$. 
Then we have the inequality:
\begin{equation*}
	\sup_{(\epsilon_0^*,\epsilon_1) \in \mathcal{E}_0^{'} \times \mathcal{E}_1} \hat{\psi} \circ g_p(\epsilon_0, \epsilon_1,\theta) \leq \sup_{(\epsilon_0,\epsilon_1) \in \mathcal{E}_0 \times \mathcal{E}_1} \hat{\psi} \circ g_p(\epsilon_0, \epsilon_1,\theta) = \sup_{\epsilon \in \mathcal{E}} \hat{\psi}(\epsilon).
\end{equation*}
The same argument can be applied to analyze the infimum:
\begin{equation*}
	\inf_{(\epsilon_0^*,\epsilon_1) \in \mathcal{E}_0^{'} \times \mathcal{E}_1} \hat{\psi} \circ g_p(\epsilon_0^*, \epsilon_1,\theta) \geq \inf_{(\epsilon_0,\epsilon_1) \in \mathcal{E}_0 \times \mathcal{E}_1} \hat{\psi} \circ g_p(\epsilon_0, \epsilon_1,\theta)= \inf_{\epsilon \in \mathcal{E}} \hat{\psi}(\epsilon).
\end{equation*}
So (i) is proved. As $W_g = \sup_{\epsilon \in \mathcal{E}} \hat{\psi}(\epsilon) - \inf_{\epsilon \in \mathcal{E}} \hat{\psi}(\epsilon)$ and $W_r = \sup_{(\epsilon_0,\epsilon_1) \in \mathcal{E}_0 \times \mathcal{E}_1} \hat{\psi} \circ g_p(\epsilon_0, \epsilon_1,\theta) - \inf_{(\epsilon_0,\epsilon_1) \in \mathcal{E}_0 \times \mathcal{E}_1} \hat{\psi} \circ g_p(\epsilon_0, \epsilon_1,\theta)$, so $W_r \leq W_g$. This completes the proof of (ii).

\subsection{Proof of Theorem 2} \label{sec:A4}

We first prove the asymptotic equivalence between the RCT-enhanced estimator \ref{eq:A4} and the synthesis estimator  where the sensitivity parameter is a function of $X$ i.e. $\epsilon = \epsilon(X)$ and $\hat{\psi}(\epsilon) = n^{-1}\sum_{i = 1}^{n}\{f \circ \epsilon(X)\} \to \psi(\epsilon) = E\{f \circ \epsilon(X)\}$ in probability, where $f$ is a known function of $\epsilon$.
\begin{equation*}
	\begin{aligned}
		\psi \circ g_p(\epsilon_0^*, \epsilon_1,\theta)
		& = E[f \circ \{\epsilon_0^*(X)I(V^* = 0) +\epsilon_1(X)I(V^* = 1)\}]  \\
		& = E\{f \circ \epsilon_0^*(X)I(V^* = 0) + f \circ \epsilon_1(X)I(V^* = 1)\}  \\
		& = E\{f \circ \epsilon_0^*(X)I(V^* = 0)\} + E\{(f \circ \epsilon_1(X)I(V^* = 1)\} \\
		& = E\{f \circ \epsilon_0^*(X)\}pr(V^* = 0) + E\{f \circ \epsilon_1(X)\}pr(V^* = 1) \\
		& = p_0 \psi_0 + p_1 \psi_1(\epsilon_1). \\
	\end{aligned}
\end{equation*}

With the same $\epsilon_1$, the RCT enhanced estimator $\hat{\psi} \circ g_p(\epsilon_0^*, \epsilon_1,\theta) \to \psi \circ g_p(\epsilon_0^*, \epsilon_1,\theta)$ in probability, and the synthesis estimator $\hat{\psi}_{syn} = p_0^N \hat{\psi}_{gen} +p_1^N \hat{\psi}_1(\epsilon_1) \to  p_0 \psi_0 + p_1 \psi_1(\epsilon_1)$ in probability. Therefore, they are asymptotically equivalent. 

For the general case (ii), for a specific coordinate $(\epsilon_0^*, \epsilon_1) \in  \mathcal{E}_0^{'} \times \mathcal{E}_1$, 
we have:
\begin{equation*}
	\hat{\psi} \circ g_p (\epsilon_0^{*},\epsilon_1; \theta) = p_0 \psi_0 + p_1 \psi_1(\epsilon_1^\prime) + o_p(1).
\end{equation*}
where $\epsilon_1^\prime$ is different from $\epsilon_1$ because the two sides use different parameter decompositions. In particular, the right hand side doesn't contain nuisance parameter $\theta$.
Since $\hat{\psi}_{gen} \to \psi_0$ and we require $\hat{\psi}_0(\epsilon_0^*) = \hat{\psi}_{gen}$, $p_0^N \to p_0$ and $p_1^N \to p_1$, the formula 
\begin{equation*}
	\frac{\hat{\psi} \circ g_p (\epsilon_0^{*},\epsilon_1; \theta) - p_0^N \hat{\psi}_{gen}}{p_1^N}
\end{equation*}
is consistent for $\psi_1(\epsilon_1^\prime)$. 
For the synthesis sensitivity analysis estimator, if we choose the sensitivity parameter as $\epsilon_1^\prime$,
we have $ \hat{\psi}_1(\epsilon_1^{'}) = \psi_1(\epsilon_1^\prime) + o_p(1)$. Therefore
\begin{equation}\label{eq:A5}
	\hat{\psi}_{syn} =p_0^N\hat{\psi}_{gen} + p_1^N\hat{\psi}_1(\epsilon_1^{'}) = \hat{\psi} \circ g_p (\epsilon_0^{*},\epsilon_1; \theta) + o_p(1),
\end{equation}
and the set of such $\epsilon_1^{'}$ is the solution set $\mathcal{E}_{sol}$. The proof ends.

\subsection{Proof of Corollary 1}
According to Equation \ref{eq:A5}, we know that if we vary $\epsilon_1^{'}$ in $\mathcal{E}_{sol}$ in the synthesis method, then we will get the same sensitivity analysis bounds as using RCT-enhanced sensitivity analysis method. As proposed in Section \ref{eq:A3}, since the RCT-enhanced estimator can get a smaller bound width than the general sensitivity analysis method with $\epsilon_1 \in \mathcal{E}_1$, it can be deduced that the synthesis method yields a narrower bound than the original sensitivity analysis method. Then (i) is proved, i.e. $W_s \leq W_g$.

As for (ii), according to Equation \ref{eq:A5}, for $\mathcal{E} = \mathcal{E}_0 \times \mathcal{E}_1$, we have:
\begin{equation*} 
	\begin{aligned}
		\frac{W_s}{W_g} & = \frac{\sup_{\epsilon_1^{'} \in \mathcal{E}_{sol}} \hat{\psi}_{syn}(\epsilon_1^{'}) - \inf_{\epsilon_1^{'}\in \mathcal{E}_{sol}} \hat{\psi}_{syn}(\epsilon_1^{'}) }{\sup_{\epsilon \in \mathcal{E}} \hat{\psi}(\epsilon) - \inf_{\epsilon \in \mathcal{E}} \hat{\psi}(\epsilon)} \\
		& = \frac{\sup_{\epsilon_1^{'} \in \mathcal{E}_{sol}} p_0^N\hat{\psi}_{gen} + p_1^N\hat{\psi}_1(\epsilon_1^{'}) - \{\inf_{\epsilon_1^{'}\in \mathcal{E}_{sol}} p_0^N\hat{\psi}_{gen} + p_1^N\hat{\psi}_1(\epsilon_1^{'})\} }{\sup_{\epsilon \in \mathcal{E}} \hat{\psi}(\epsilon) - \inf_{\epsilon \in \mathcal{E}} \hat{\psi}(\epsilon)} \\
		& = \frac{\sup_{(\epsilon_0^*,\epsilon_1) \in \mathcal{E}_0^{'} \times \mathcal{E}_{1}} \hat{\psi}\{g_p(\epsilon_0^{'}, \epsilon_1; \theta)\} - \inf_{(\epsilon_0^*,\epsilon_1) \in \mathcal{E}_0^{'} \times \mathcal{E}_{1}} \hat{\psi}\{g_p(\epsilon_0^{'}, \epsilon_1; \theta)\} + o_p(1) }{\sup_{(\epsilon_0,\epsilon_1) \in \mathcal{E}_0 \times \mathcal{E}_{1}} \hat{\psi}\{g_p(\epsilon_0, \epsilon_1; \theta)\} - \inf_{(\epsilon_0,\epsilon_1) \in \mathcal{E}_0 \times \mathcal{E}_{1}} \hat{\psi}\{g_p(\epsilon_0, \epsilon_1; \theta)\} }. \\
	\end{aligned}
\end{equation*}

Therefore, as $p_0 \to 1$, $p_0^N \to 1$ due to the consistency. Then we have the limitation due to the second equation above:
\begin{equation*}
	\begin{aligned}
		\lim_{p_0 \to 1} \frac{W_s}{W_g} & = \lim_{p_0^N \to 1}\frac{\sup_{\epsilon_1^{'} \in \mathcal{E}_{sol}} p_0^N\hat{\psi}_{gen} + p_1^N\hat{\psi}_1(\epsilon_1^{'}) - \{\inf_{\epsilon_1^{'}\in \mathcal{E}_{sol}} p_0^N\hat{\psi}_{gen} + p_1^N\hat{\psi}_1(\epsilon_1^{'})\} }{\sup_{\epsilon \in \mathcal{E}} \hat{\psi}(\epsilon) - \inf_{\epsilon \in \mathcal{E}} \hat{\psi}(\epsilon)} \\
		& = \frac{\hat{\psi}_{gen} - \hat{\psi}_{gen}}{\sup_{\epsilon \in \mathcal{E}} \hat{\psi}(\epsilon) - \inf_{\epsilon \in \mathcal{E}} \hat{\psi}(\epsilon)} \\
		&= 0.
	\end{aligned}
\end{equation*}

As $p_0 \to 0$, $p_0^N \to 0$. According to the third equation, we have:
\begin{equation*}
	\begin{aligned}
		\lim_{p_0 \to 0} \frac{W_s}{W_g} & =\lim_{p_0 \to 0} \frac{\sup_{(\epsilon_0^*,\epsilon_1) \in \mathcal{E}_0^{'} \times \mathcal{E}_{1}} \hat{\psi}\{g_p(\epsilon_0^{'}, \epsilon_1; \theta)\} - \inf_{(\epsilon_0^*,\epsilon_1) \in \mathcal{E}_0^{'} \times \mathcal{E}_{1}} \hat{\psi}\{g_p(\epsilon_0^{'}, \epsilon_1; \theta)\} +o_p(1)}{\sup_{(\epsilon_0,\epsilon_1) \in \mathcal{E}_0 \times \mathcal{E}_{1}} \hat{\psi}\{g_p(\epsilon_0, \epsilon_1; \theta)\} - \inf_{(\epsilon_0,\epsilon_1) \in \mathcal{E}_0 \times \mathcal{E}_{1}} \hat{\psi}\{g_p(\epsilon_0, \epsilon_1; \theta)\} } \\
		& = \frac{\sup_{\epsilon_1 \in \mathcal{E}_1} \hat{\psi}(\epsilon_1) - \inf_{\epsilon_1 \in \mathcal{E}_1} \hat{\psi}(\epsilon_1) +o_p(1)}{\sup_{\epsilon_1 \in \mathcal{E}_1} \hat{\psi}(\epsilon_1) - \inf_{\epsilon_1 \in \mathcal{E}_1} \hat{\psi}(\epsilon_1)} \\
		& \to \frac{\sup_{\epsilon_1 \in \mathcal{E}_1} \hat{\psi}(\epsilon_1) - \inf_{\epsilon_1 \in \mathcal{E}_1} \hat{\psi}(\epsilon_1) }{\sup_{\epsilon_1 \in \mathcal{E}_1} \hat{\psi}(\epsilon_1) - \inf_{\epsilon_1 \in \mathcal{E}_1} \hat{\psi}(\epsilon_1)} \\
		& = 1,
	\end{aligned}
\end{equation*}
where the second equation holds due to the fact of the transformation function proposed in the paper that
\begin{equation*}
	\lim_{p_j \to 1} \epsilon = \lim_{p_j \to 1} g_p(\epsilon_0, \epsilon_1; \theta) = \epsilon_j, \quad j \in \{0,1\}.
\end{equation*}

\newpage

\section{Supplemental Results}\label{sec:B}
\subsection{Simulation Results}

This subsection includes additional simulation results comparing sensitivity analysis bound widths in scenario II. For a fixed 
$q = 0.7$, Table \ref{tab:B1} presents two bounding approaches: the observed covariate benchmark bounding and the fixed value bounding method, along with the corresponding sensitivity analysis bound widths for both sensitivity analysis methods. Figure \ref{fig:B1} illustrates the trend of the bound width ratio of the synthesis method over the general sensitivity method as the RCT-overlapping proportion $q$ varies.

\begin{table}[h]
	\centering
	\begin{threeparttable}
		\caption{Sensitivity analysis bounds comparison between synthesis method and original method based on omitted variable bias framework in \citet{cinelli2020making}}
		\label{tab:B1}
	\begin{tabular}{cccccc}	
		&	& OS & OS+RCT(OM) & OS+RCT(IPSW) & OS+RCT(AIPSW)  \\
		& lb& 9.37(1.18)&9.00(0.78) & 9.14(0.82)& 8.83(0.82) \\
		$X_1$ bounding & ub &16.55(1.10) &13.35(0.73) & 13.55(0.74)&13.36(0.71)  \\
		& mbw &7.18 &4.35 &4.40 & 4.53  \\
		& lb &6.94(0.75) &8.47(0.40) & 8.86(0.52)&8.64(0.38) \\
		fixed value bounding & ub &19.14(0.75) &13.61(0.44) &13.81(0.57) & 13.58(0.41) \\
		& mbw &12.20 &5.15 &4.95 & 4.94\\	
	\end{tabular}
	\end{threeparttable}
	\begin{tablenotes}
		\item 	lb, lower bound; ub, upper bound; mbw, mean bound width. Other abbreviations refer to the paper.
	\end{tablenotes}
\end{table}

\begin{figure}[h]
	\centering
	\includegraphics[width=0.6\linewidth]{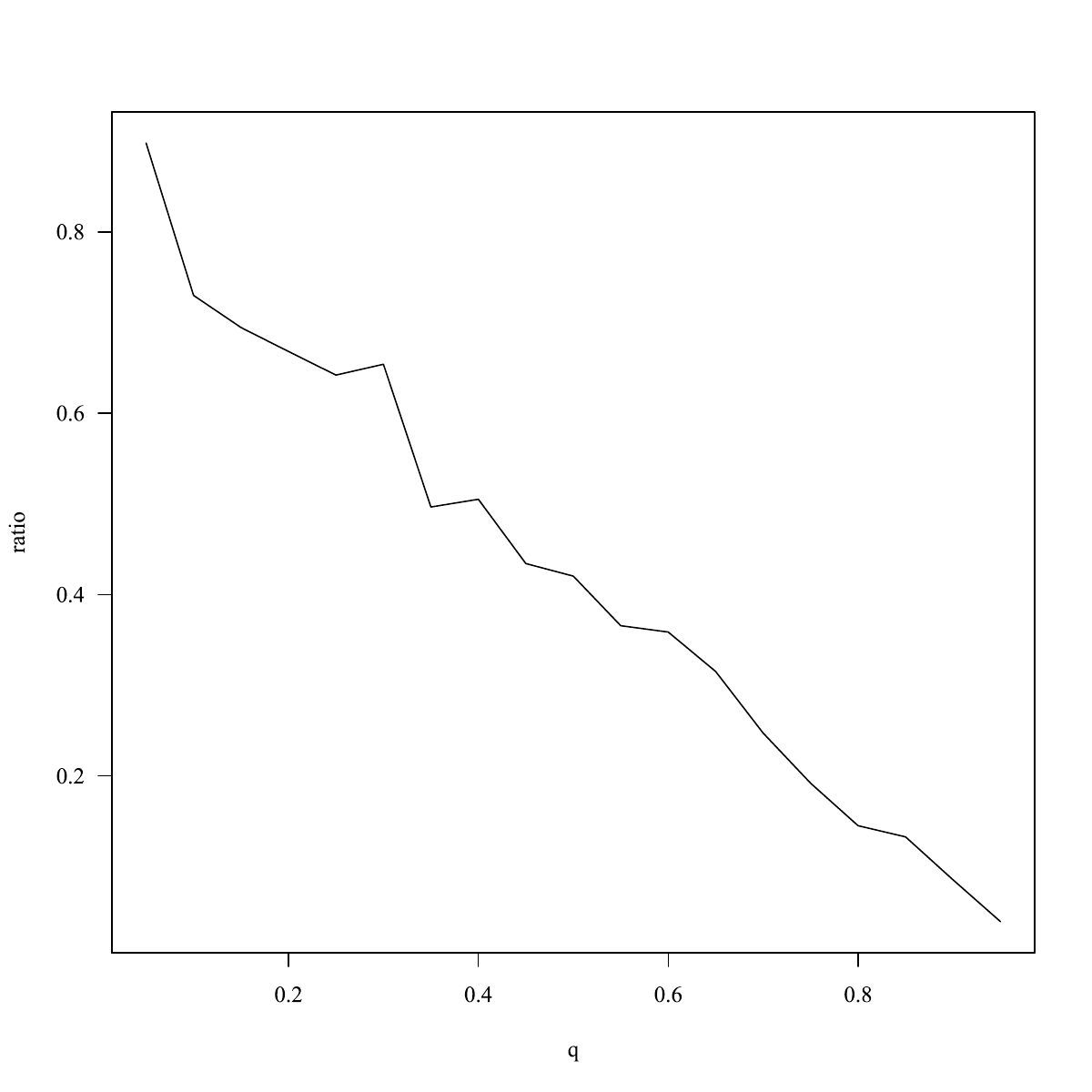}
	\caption{Sensitivity analysis bound width ratio width a varing RCT-overlap proportion.}
	\label{fig:B1}
\end{figure}

\newpage
\subsection{Real-data Analysis Results}
This subsection presents further analysis results comparing the effectiveness of SXC and Losartan in reducing hypertension using the synthesis method. Tables \ref{tab:B2} and \ref{tab:B3} display the covariate benchmarking bounds for unmeasured confounding strength, with the benchmarking methodology referencing \citet{cinelli2020making}. Table \ref{tab:B4} and Figure \ref{fig:B2} provide additional comparison results.

For DDBP at week 8, we conservatively use BDBP as the benchmark variable for bounding the confounding strength.
Both methods conclude that there is no significant difference in the efficacy of SXC and Losartan on reducing DSBP, considering unmeasured confounders conservatively. However, the synthesis method offers a narrower sensitivity analysis bound.

\begin{table}[h]
	\centering
	\begin{threeparttable}
		\caption{Formal benchmarking for the confounding strength in relationship between the treatment and DSBP at week 8}
		\label{tab:B2}
			\begin{tabular}{ccccc}
			& \multicolumn{2}{c}{Entire OS}	&  \multicolumn{2}{c}{RCT non-overlapping OS}\\
			Bounding variable & $R^2_{A \sim U\mid X}$ &$R^2_{Y \sim U\mid A,X}$ & $R^2_{A \sim U\mid X,V^* = 1}$ &$R^2_{Y \sim U\mid A,X,V^* = 1}$ \\
			age	&3.12$\times 10^{-3}$& 5.36 $\times 10^{-3}$ &6.02 $\times 10^{-4}$  & 2.54 $\times 10^{-3}$ \\
			BSBP & 4.14 $\times 10^{-2}$ & 9.57 $\times 10^{-1}$ &3.62 $\times 10^{-2}$ & 9.78 $\times 10^{-1}$ \\
			BDBP & 1.38$\times 10^{-2}$ & 1.31$\times 10^{-3}$&1.56 $\times 10^{-2}$ &9.16 $\times 10^{-4}$ \\
			BMI & 3.79$\times 10^{-3}$ & 1.67$\times 10^{-3}$&4.78 $\times 10^{-3}$ &4.86 $\times 10^{-3}$ \\
			gender & 1.39$\times 10^{-6}$ & 3.13$\times 10^{-4}$&8.87 $\times 10^{-7}$ &6.85 $\times 10^{-4}$ \\
			smoke & $1.44\times 10^{-4}$ & 7.90$\times 10^{-6}$&2.00 $\times 10^{-6}$ &3.41 $\times 10^{-4}$ \\
			marriage & 1.47$\times 10^{-6}$ & 2.82$\times 10^{-5}$&1.58 $\times 10^{-4}$ &2.50 $\times 10^{-4}$ \\
		\end{tabular}
	\end{threeparttable}
	\begin{tablenotes}
		\item
	\end{tablenotes}
\end{table}

\begin{table}[h]
	\centering
	\begin{threeparttable}
		\caption{Formal benchmarking for the confounding strength in relationship between the treatment and DDBP at week 8}
		\label{tab:B3}
		\begin{tabular}{ccccc}
			&  \multicolumn{2}{c}{Entire OS}  &  \multicolumn{2}{c}{RCT non-overlapped OS} \\
			Bounding variable & $R^2_{A \sim U\mid X}$ &$R^2_{Y \sim U\mid A,X}$ & $R^2_{A \sim U\mid X,V^* = 1}$ &$R^2_{Y \sim U\mid A,X,V^* = 1}$ \\
			age	&3.12$\times 10^{-3}$& 3.15 $\times 10^{-6}$ &6.03 $\times 10^{-4}$  & 8.78 $\times 10^{-5}$ \\
			BSBP & 4.14 $\times 10^{-2}$ & 7.97 $\times 10^{-3}$ &3.62 $\times 10^{-2}$ & 6.04 $\times 10^{-3}$ \\
			BDBP & 1.38$\times 10^{-2}$ & 8.51$\times 10^{-1}$&1.56 $\times 10^{-2}$ &8.79 $\times 10^{-1}$ \\
			BMI & 3.79$\times 10^{-3}$ & 9.52$\times 10^{-4}$&4.78 $\times 10^{-3}$ &3.01 $\times 10^{-3}$ \\
			gender & 1.39$\times 10^{-6}$ & 1.68$\times 10^{-5}$&8.67 $\times 10^{-7}$ &2.00 $\times 10^{-4}$ \\
			smoke & $1.44\times 10^{-4}$ & 2.85$\times 10^{-4}$&2.00 $\times 10^{-6}$ &9.26 $\times 10^{-4}$ \\
			marriage & 1.47$\times 10^{-6}$ & 1.16$\times 10^{-3}$&1.58 $\times 10^{-4}$ &2.29 $\times 10^{-3}$ \\
		\end{tabular}
	\end{threeparttable}
	\begin{tablenotes}
		\item
	\end{tablenotes}
\end{table}

\begin{table}[h]
	\centering
	\begin{threeparttable}
		\caption{Comparison of sensitivity analysis bounds between general sensitivity analysis method in \citet{cinelli2020making} and synthesis method (RCT generalization method: AIPSW)}
		\label{tab:B4}
	 \begin{tabular}{ccccccc}
		& \multicolumn{3}{c}{OS} & \multicolumn{3}{c}{OS+RCT} \\
		Outcome  & lb  & ub & mbw & lb & ub & mbw \\
		DSBP    & -3.55([-4.42,-2.62]) & 4.96([4.07,5.94]) &8.51 &-2.60([-3.96,-1.44]) & 3.31([2.04,4.47]) & 5.92\\
		DDBP  &-0.09([-0.62,0.60]) &2.96([2.27,3.62])&3.05 & 0.59([-0.38,1.43]) & 3.05([2.36,3.64]) & 2.47  \\
	\end{tabular}
	\end{threeparttable}
	\begin{tablenotes}
		\item
	\end{tablenotes}
\end{table}

\begin{figure}[h]
	\centering
	\subfigure{\includegraphics[width=.4\textwidth]{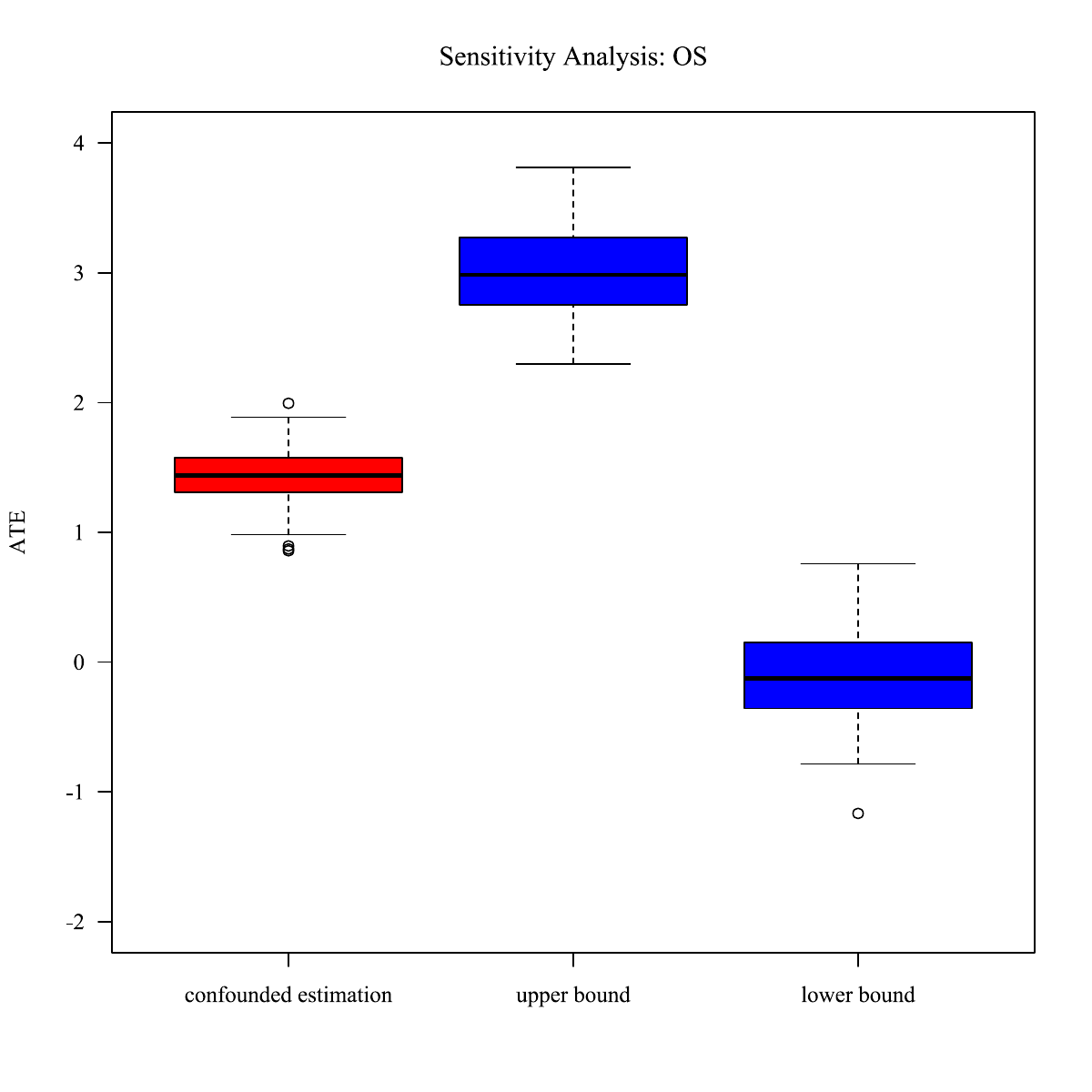}}
	\subfigure{\includegraphics[width=.4\textwidth]{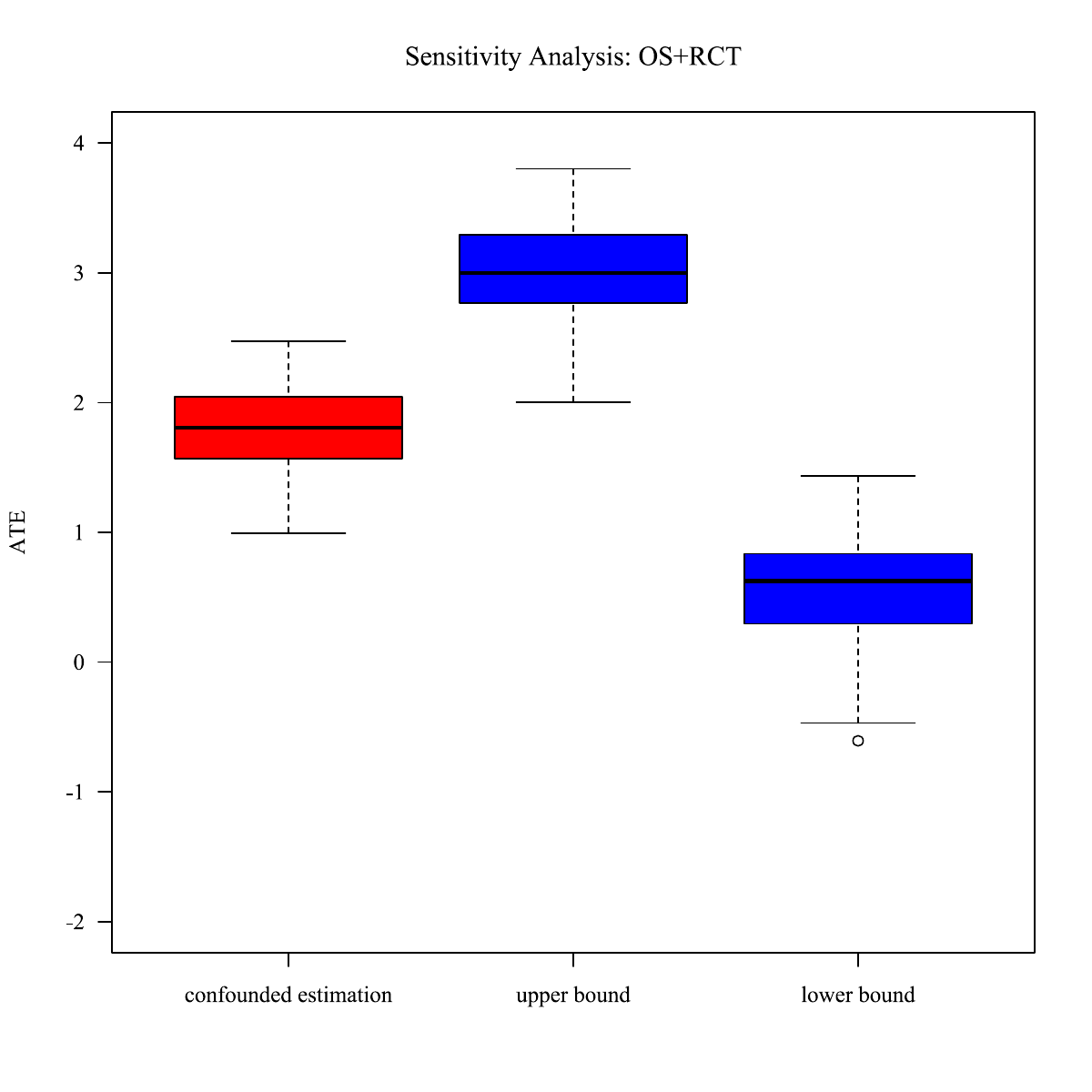}}
	\caption{Comparison of sensitivity analysis bounds between general sensitivity analysis method in \citet{cinelli2020making} and synthesis method (RCT generalization method: AIPSW). Outcome: DDBP at week 8.}
	\label{fig:B2}
\end{figure}

\end{document}